\documentclass[runningheads]{llncs}
\usepackage[T1]{fontenc}
\usepackage{graphicx}
\usepackage{color}

\usepackage{amsmath}
\usepackage{amsfonts}
\usepackage{amssymb}

\usepackage{quantikz}
\usepackage[linesnumbered,ruled]{algorithm2e}
\usepackage{algpseudocode}
\usepackage[%
  colorlinks=true,
  citecolor=blue,
  pagebackref=true
]{hyperref}
\usepackage{breakcites}

\usepackage{geometry}
\newcommand{\N}{\mathbb{N}}

\newcommand{\bit}{\{0,1\}}

\newcommand{\algo}{\mathcal}

\newcommand{\Sp}{\mathsf{Sp}}
\newcommand{\priv}[1]{\mathsf{#1}^{\mathrm{priv}}}
\newcommand{\pub}[1]{\mathsf{#1}^{\mathrm{pub}}}
\newcommand{\pubpriv}[1]{\mathsf{#1}}
\newcommand{\indiffparams}{(S, T, S_{\mathsf{sim}}, T_{\mathsf{sim}}, \epsilon)}
\usepackage{tikz}

\begin{document}
\title{(Quantum) Indifferentiability and\\ Pre-Computation}

\author{ }
\institute{ }
\authorrunning{ }
\titlerunning{ }
%
%
\author{Joseph Carolan\inst{1} \and
Alexander Poremba\inst{2}\and
Mark Zhandry\inst{3}}
\institute{University of Maryland. \email{jcarolan@umd.edu} \and
Massachusetts Institute of Technology. \email{poremba@mit.edu}
\and
NTT Research.
\email{mzhandry@gmail.com}}
\maketitle              
\begin{abstract}
Indifferentiability is a popular cryptographic paradigm for analyzing the security of ideal objects---both in a classical as well as in a quantum world. It is typically stated in the form of a composable and simulation-based definition, and captures what it means for a construction (e.g., a cryptographic hash function) to be ``as good as'' an ideal object (e.g., a random oracle). Despite its strength, indifferentiability is not known to offer security against \emph{pre-processing} attacks in which the adversary gains access to (classical or quantum) advice that is relevant to the particular construction. 
In this work, we show that indifferentiability is (generically) insufficient for capturing pre-computation. To accommodate this shortcoming, we propose a strengthening of indifferentiability which is not only composable but also takes arbitrary pre-computation into account. As an application, we show that the one-round sponge is indifferentiable (with pre-computation) from a random oracle. This yields the first (and tight) classical/quantum space-time trade-off for one-round sponge inversion.

\end{abstract}
\section{Introduction}
Hash functions are fundamental objects in cryptography, used in a multitude of applications such as password storage, integrity checks, and in digital signature schemes. Many real world cryptographic schemes can only be proven secure in the random oracle model \cite{bellare1993random} (ROM), in which a hash function is instead treated like an idealized perfectly random function. In this model, adversaries receieve only black-box access to the hash function, which enables one to prove query lower bounds justifying the security of a construction. Similarly, the quantum random oracle model \cite{boneh2011random} (QROM) models adversaries as having quantum query-access to an idealized random function. These tools have since become indispensable in analyzing real world cryptographic systems, both in the (post-)quantum and the classical setting.

\paragraph{Indifferentiability.}

In the real world, however, hash functions are built from lower-level building blocks, such as compression functions or publicly invertible permutations. The structure of these hashes can lead to attacks: length-extension attacks on Merkle-Damg\r{a}rd are a famous example; another example is circular-secure encryption when using Davies-Meyer~\cite{10.1145/1315245.1315303}. These attacks work regardless of the lower-level building block, and simply exploit the way the building block is used in the higher-level protocol.

One possibility is to analyze a given hash construction when used in specific scenarios. A much better solution~\cite{coron2005merkle} is to ensure that the hash function is \textit{indifferentiable} from a random oracle. Indifferentiability was first defined by Maurer, Renner, and Holenstein \cite{maurer04indiff}, and is a composable, simulation-based definition. An indifferentiable hash function is ``as good as'' a random oracle, in that we can ``lift'' any single-stage security property of random oracles---which capture most of the standard properties---to conclude that the property also holds for an indifferentiable hash function, provided the underlying building block is modeled as an idealized object. Indifferentiability therefore ensures that no attacks were introduced in the conversion from the lower-level building block to the higher-level hash function. Therefore, rather than analyzing the hash function in every scenario of interest, we can simply prove that it is indifferentiable and immediately conclude its security.

A strengthening of this notion, called reset indifferentiability, was introduced by Ristenpart, Shacham, and Shrimpton \cite{ristenpart11careful}. This notion requires a stateless simulator, but allows composable security in games with an arbitrary number of stages and adversaries. While numerous positive results are known for plain indifferentiability (such as domain extension of random oracles, and equivalence between ideal ciphers and random oracles \cite{coron2008randomcipher,holenstein2011cipherrevisited,dai2016feistel}), various barriers apply to reset indifferentiability \cite{ristenpart11careful,luykx2012impossibility,demay2013resource,baecher2013reset}: in particular, domain extension is not possible. Despite this barrier, Zhandry \cite{Zhandry21} has used reset indifferentiability to show, among other things, that ideal ciphers imply fixed size random oracles. In particular, it is shown that the single-round sponge is weakly reset-indifferentiable (quantum or classical) from a random oracle, when the rate does not exceed the capacity.

\paragraph{Adversaries with pre-computation.}

A common and desirable property of a cryptographic scheme is security against adversaries with some pre-computed advice. In the context of the random oracle model and other idealized models, this is usually considered in the auxiliary input model introduced by Unruh~\cite{cryptoeprint:2007/168}. In this model, adversaries are split into an inefficient offline and efficient online stage, where only a single (potentially small) advice message can be passed from offline to online. The online adversary then receives a challenge, or more generally must win some security game with the help of the advice. Many classical and quantum results are known  in this model \cite{coretti2018nonuniform,yao1990coherent,de2010timespace,hellman1980cryptanalytic,fiat2000rigorous,corrigan2019function,hhan2019auxiliary,chung2020quantumadvice,chung2020tight}.

The aforementioned works assume a random oracle. As mentioned above, however, hash functions are typically built from lower-level building blocks, and this structure may be exploited in attacks. A recent line of work therefore has investigated pre-computation attacks on structured hash functions, giving both attacks and lower-bounds in different settings~\cite{coretti2018nonuniform,ACDW20,GK22,freitag22spongetradeoffs,akshima23spongesttradeoffs,Akshima24}. There are not, to our knowledge, any known space-time tradeoffs (either classical or quantum) for inverting the one-round sponge.
A natural question is:

\begin{center}
\emph{Why not just use the stong notion of indifferentiability to lift space-time trade-offs for random oracles to structured hash functions?}
\end{center}

 After all, the goal of indifferentiability is to avoid having to analyze a structured hash construction in every conceivable scenario, and instead simply lift existing random oracle results. The short answer, unfortunately, is that indifferentiability as currently defined simply does not work. Due to the pre-computation phase, space-time trade-offs are not single-stage games, and since the pre-computation is inefficient, it is also not a multi-stage game. Therefore, the lifting theorems for (reset) indifferentiability simply do not apply to space-time trade-offs. In fact, we show that this is inherent: there is a function which is (strongly, statistical) reset indifferentiable from a random oracle, but admits a pre-computation attack on function inversion with polynomial-size advice and polynomial computation. This leads us to ask the following question:
\begin{center}
\emph{Does this mean that indifferentiability cannot help us understand space-time tradeoffs for structured hash functions?}
\end{center}

\section{Our contributions}

We now give an overview of our main results.

\subsection{(Quantum) indifferentiability with pre-computation}
Motivated by the aforementioned counterexample, we introduce a notion of indifferentiability with pre-computation---a strengthening of plain indifferentiability, in Section \ref{section:indiff-with-pre}. In this definition, a distinguisher is split into an offline and an online part, where the offline part is allowed unbounded access to some primitive. However, only a limited size message can be passed to the online adversary, which is bounded. We define indifferentiability by way of a simulator which is also split into two stages, an inefficient offline and an efficient online simulator, and again allow the simulators to pass some bounded size advice from offline to online. As our goal is to show (tight) space-time tradeoffs, rather than restricting to efficient (polynomial time) adversaries we instead phrase our definitions in a more fine-grained manner, parameterized by advice sizes, query count, and success probabilities. In strong indifferentiability, the simulator must simulate the adversaries' interface in both the offline and online phase, though in weak indifferentiability we instead allow the simulator to prepare both the advice for the online adversary, as well as for the online simulator.

In Section \ref{sec:composition}, we introduce a composition theorem under our definition of indifferentiability with pre-computation. Informally, this composition theorem says that if construction $\pubpriv C$ is indifferentiable from construction $\pubpriv R$, then a security game with pre-computation (for instance a space-time tradeoff for function inversion) instantiated using $\pubpriv C$ is as secure as one instantiated using $\pubpriv R$. This statement holds up to some loss incurred from indifferentiability, which depends on how much advice the simulator needs as well as how many queries.
\begin{theorem}[Informal version of Theorem \ref{thm:composition}]
    If construction $\pubpriv C$ is indifferentiable with pre-computation from construction $\pubpriv R$, then any security game with a pre-computing adversary which is secure when instantiated with $\pubpriv R$ remains secure when instantiated with $\pubpriv C$.
\end{theorem}

\subsection{(Quantum) space-time trade-offs for sponge inversion}

In recent years, the National Institute of Standards and Technology (NIST) announced a new international hash function standard known as SHA-3. Unlike its predecessor SHA-2, which was rooted in the Merkle-Damg\aa rd construction~\cite{Mer88,Mer90,eurocrypt-1987-2247}, the new hash function standard uses Keccak~\cite{KeccakSub3}---a family of cryptographic functions based on the idea of \emph{sponge hashing}~\cite{KeccakSponge3}. This particular approach  allows for both variable input length and variable output length, which makes it particularly attractive towards the design of cryptographic hash functions. 
The internal state of a sponge function gets updated through successive applications of a so-called \emph{block function} $\varphi: \bit^{r+c} \rightarrow \bit^{r+c}$ (which is typically modeled as a random permutation), where we call the parameters $r \in \N$ the \emph{rate} and $c \in \N$ the \emph{capacity} of the sponge. 

\paragraph{One-round sponge.} 
Suppose that $\varphi: \bit^{r+c} \rightarrow \bit^{r+c}$ is a permutation. In the special case when there is only a single application of the block function, the sponge function $\mathsf{Sp}^\varphi: \bit^{r} \rightarrow \bit^r$ takes a simple form which is illustrated in Figure~\ref{fig:single-sponge}; namely, on input $x \in \bit^r$, the output is given by $y = \mathsf{Sp}^\varphi(x)$, where $y$ corresponds to the first $r$ bits of $\varphi(x||0^c)$. In other words, $\mathsf{Sp}^\varphi$ is defined as the restriction of $\varphi$ onto the first $r$ bits of its output.
 
\begin{figure}[h]
\begin{center}
{\small
\begin{tikzpicture}
  \draw (5,-0.75) rectangle (6,0.75) node [pos=.5]{$\varphi$}; 

\draw[-] (4.6,0.35) node[left]{$x$ \hspace{1mm}} -- (5,0.35);
\draw[-] (4.6,-0.35) node[left]{$0^c$} --(5,-0.35);
  \draw[-] (6.0,0.35) node[right]{\hspace{5mm}$y$} --(6.4,0.35);
\draw[-] (6.0,-0.35) node[right]{\hspace{5mm}$z$} --(6.4,-0.35);
 \end{tikzpicture}
}
\end{center}
\caption{The one-round sponge.}
\label{fig:single-sponge}
\end{figure}
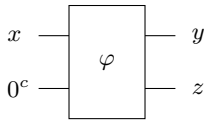

As an application of our results on indifferentiability, we show that the one-round sponge construction is both quantumly and classically indifferentiable with pre-computation from a random oracle $f: \bit^r \rightarrow \bit^r$ when the rate does not exceed the capacity (see Section \ref{sec:sponge-indiff}). Our proof consists of two parts. First,  in Section \ref{sec:improved-sym}, we use \emph{symmetrization techniques} and give a (strong, stateless) simulator with shared randomness which generates a permutation whose sponge-hash precisely matches a given random function. Next, we show how to remove the shared randomness at the cost of downgrading to a weak and stateful simulator with a single bit of advice.
\begin{theorem}[Informal version of Theorem \ref{thm:sponge-indiff} and Corollary \ref{cor:weak-indiff-sponge}]
    The single-round sponge is indifferentiable with pre-computation from a random oracle, both quantumly and classically, when the rate does not exceed the capacity. This holds with unbounded adversaries, either with stateless strong simulators with shared randomness, or stateful weak simulators that pass a single bit of advice.
\end{theorem}

With shared randomness removed, we show in Section \ref{sec:ST-trade-off} how to derive tight quantum space-time tradeoffs for inverting the single-round sponge. To our knowledge, this is the first classical \emph{or} quantum space-time tradeoff for sponge inversion. We summarize our results in Table \ref{tab:permutation inversion work}.

\begin{table}[t]
    \centering
    \caption{Summary of our space-time trade-offs in Section \ref{sec:ST-trade-off}.}\label{tab:permutation inversion work}
    \begin{tabular}{|c|p{3.2cm}|p{3.2cm}|}
         \hline
   &  Function inversion &  Sponge inversion\\
   \hline
    Classical advice, & $ST = \tilde{\Omega}(\epsilon \,2^r)$ &$ST = \tilde{\Omega}(\epsilon \, 2^r)$ \\
    classical queries& Refs.~\cite{yao1990coherent,de2010timespace} &(\textbf{this work})\\
    \hline
    Classical advice, & $ST + T^2= \tilde{\Omega}(\epsilon \,2^r)$ & $ST + T^2= \tilde{\Omega}(\epsilon \,2^r)$\\
    quantum queries&Ref.~\cite{chung2020tight} & (\textbf{this work})\\
    \hline
    Quantum advice, & $ST + T^2= \tilde{\Omega}(\epsilon^3 \,2^r)$ &  $ST + T^2= \tilde{\Omega}(\epsilon^3\, 2^r)$ \\ 
    quantum queries &Ref.~\cite{chung2020tight} & (\textbf{this work})\\
    \hline
    \end{tabular}
\end{table}

\subsection{Related work}
\label{sec-related}

We now briefly discuss several related works on the topic of both (quantum) indifferentiability, pre-computation and the sponge construction.

Maurer, Renner, and Holenstein \cite{maurer04indiff} first proposed the notion of \emph{indifferentiability} as a composable and simulation-based definition for what it
means for a construction to be ``as good as'' as
an ideal object. Ristenpart, Shacham, and Shrimpton~\cite{ristenpart11careful} observed that indifferentiability is insufficient for
“multi-stage” games, and proposed the notion of \emph{reset indifferentiability} instead which requires the simulator to be stateless.
Bertoni, Daemen, Peeters and Van Assche~\cite{10.1007/978-3-540-78967-311} proved the indifferentiability of the many-round sponge construction.
Carstens, Ebrahimi, Tabia, and Unruh~\cite{cryptoeprint:2018/257} initiated the study of indifferentiability in the quantum setting, and analyzed the security of both Feistel networks and the sponge construction under conjectures.
Zhandry \cite{Zhandry2018} introduced the compressed oracle technique, and used it to prove quantum indifferentiability of the Merkle-Damg\r{a}rd construction.
Czajkowski, Majenz, Schaffner and Zur~\cite{Czajkowski2019} proved the quantum indifferentiability of the (many-round) sponge construction in the case when the block function is modeled as a random function or a random (non-invertible) permutation. 
Zhandry~\cite{Zhandry21} showed that the one-round sponge (in the special case when the message length is roughly half the block length) is quantumly reset-indifferentiable from a random oracle (even if the adversary has access to the inverse of the permutation). However, contrary to our work, none of the aforementioned works on indifferentiability take pre-computation into account.

Yao \cite{yao1990coherent} and De, Trevisan and Tulsiani~\cite{de2010timespace} gave (classical) space-time trade-offs for function inversion. Unruh~\cite{cryptoeprint:2007/168} introduced the auxiliary-input random oracle model. Nayebi, Aaronson, Belovs and Trevisan~\cite{10.5555/2871350.2871351} generalized space-time trade-offs for function inversion against quantum adversaries with classical advice. Later, Chung, Liao and Qian~\cite{chung2020quantumadvice} generalized these bounds in the case of quantum advice. Hhan, Xagawa and Yamakawa~\cite{hhan2019auxiliary} gave space-time trade-offs for function (and permutation) inversion in the auxiliary-input quantum random oracle model.
Chung, Guo, Liu and Qian~\cite{chung2020tight} gave the first tight quantum space-time trade-off for function inversion with both classical and quantum advice.
Alagic, Bai, Poremba and Shi~\cite{alagic2023twosided} showed quantum space-time trade-offs for two-sided permutation inversion---the task of inverting a random but invertible permutation, where the inverter also has access to a punctured inverse oracle.
Freitag, Ghoshal and Komargodski~\cite{freitag22spongetradeoffs}, and subsequently also Akshima, Duan, Guo and Liu~\cite{akshima23spongesttradeoffs,Akshima24}, gave space-time trade-offs for finding short collisions in the sponge construction. Carolan and Poremba~\cite{carolan2024oneway} gave a (tight) quantum query lower bound for one-round sponge-inversion via symmetrization techniques. In concurrent work, Majenz, Malavolta and Walter~\cite{majenz2024permutationsuperpositionoraclesquantum} also gave  (non-tight) quantum query lower bounds for the task of sponge inversion (in a more general setting) via compressed oracle techniques. Ananth, Mutreja and Poremba~\cite{cryptoeprint:2024/1687} recently gave space-time trade-offs for a simple query problem; namely that of finding elements in (one-round) sponge hash tables.
Notably, none of the aforementioned works on function inversion result in (either classical or quantum) space-time trade-offs for the single-round sponge inversion task, as in our work.

\subsection*{Acknowledgements}

The authors would like to thank Christian Majenz, Giulio Malavolta and Gorjan Alagic for useful discussions. JC is supported by the US Department of Energy grant no. DESC0020264. AP is supported by the National Science Foundation (NSF) under Grant No.\ CCF-1729369.

\section{Preliminaries}

\paragraph{Basic notation.}
For $N\in \N$, we use $[N] = \{1,2,\dots,N\}$ to denote the set of integers up to $N$. The symmetric group on $[N]$ is denoted by $S_N$. 
In slight abuse of notation, we oftentimes identify elements $x \in [N]$ with bit strings $x \in \bit^n$ via their binary representation whenever $N=2^n$ and $n \in \N$. Similarly, we identify permutations $\pi \in S_N$ with permutations $\pi: \bit^{n} \rightarrow 
\bit^n$ over bit strings of length $n$.

\paragraph{Quantum computing.} A finite-dimensional complex Hilbert space is denoted by $\algo H$, and we use subscripts to distinguish between different systems (or registers); for example, we let $\algo H_{A}$ be the Hilbert space corresponding to a system $A$. 
The tensor product of two Hilbert spaces $\algo H_{A}$ and $\algo H_{B}$ is another Hilbert space which we denote by $\algo H_{AB} = \algo H_{A} \otimes \algo H_{B}$.  We let $\algo L(\algo H)$ denote the set of linear operators over $\algo H$. A quantum system over the $2$-dimensional Hilbert space $\algo H = \mathbb{C}^2$ is called a \emph{qubit}. For $n \in \mathbb{N}$, we refer to quantum registers over the Hilbert space $\algo H = \big(\mathbb{C}^2\big)^{\otimes n}$ as $n$-qubit states. We use the word \emph{quantum state} to refer to both pure states (unit vectors $\ket{\psi} \in \algo H$) and density matrices $\rho \in \algo D(\algo H)$, where we use the notation $\algo D(\algo H)$ to refer to the space of positive semidefinite linear operators of unit trace acting on $\algo H$. A \emph{unitary} $U: \algo L (\algo H_{A}) \to \algo L(\algo H_{A})$ is a linear operator such that $U^\dagger U = U U^\dagger = I_A$, where the operator $I_A$ denotes the identity operator on system $\algo H_{A}$.
A quantum algorithm is a uniform family of quantum circuits $\{\algo A_\lambda\}_{\lambda \in \mathbb{N}}$, where each circuit $\algo A_\lambda$ is described by a sequence of unitary gates and measurements; moreover, for each $\lambda \in \mathbb{N}$, there exists a deterministic Turing machine that, on input $1^\lambda$, outputs a circuit description of $\algo A_\lambda$.
We say that a quantum algorithm $\mathcal{A}$ has oracle access to a classical function $f: \{0,1 \}^{n} \rightarrow \{0,1 \}^m$, denoted by $\mathcal{A}^f$, if $\mathcal{A}$ is allowed to use a unitary gate $O^f$ at unit cost in time. The unitary $O^f$ acts as follows on the computational basis states of a Hilbert space $\mathcal{H}_X \otimes \mathcal{H}_Y$ of $n+m$ qubits:
$$
O^f: \quad
\ket{x}_X \otimes \ket{y}_Y \longrightarrow \ket{x}_X \otimes \ket{y \oplus f(x)}_Y,
$$
where the operation $\oplus$ denotes bit-wise addition modulo $2$. Oracles with quantum query-access have been studied extensively, for example in the context of quantum complexity theory~\cite{Bennett1997}, as well as in cryptography~\cite{boneh2011random,cryptoeprint:2018/904,cryptography4010010}.

\subsection{(Quantum) Indifferentiability}

Our notation is based on~\cite{Czajkowski2019} which is close to the original definition of Maurer, Renner, and
Holenstein~\cite{maurer04indiff}. The basic idea is that an adversary who is interacting with some cryptographic system $\mathsf{C}$ has access to two \emph{interfaces}:
\begin{itemize}
\item a \emph{public} interface $\mathsf{C}_\lambda^{\mathrm{pub}}$  (for example, a permutation $\varphi$) which is some public interface that takes as input a certain number of (qu)bits, and outputs a number of (qu)bits. 

 \item a \emph{private} interface $\mathsf{C}_\lambda^{\mathrm{priv}}$ (for example, the sponge hash $\mathsf{Sp}^\varphi$ which uses the permutation $\varphi$ internally) which is also some function that takes as input a certain number of (qu)bits and outputs a number of (qu)bits. 
\end{itemize}

For the purposes of indifferentiability, we will consider constructions where the private interface is constructed from the public interface (i.e. $\priv C_\lambda [\pub C_\lambda]$ is an efficient algorithm with an oracle for $\pub C_\lambda$). We can now define what it means for two interfaces to be indifferentiable.

\begin{definition}[Indifferentiability]

Let $\lambda \in \N$ be the security parameter.
A cryptographic system $\mathsf{C}$ is $(T,\epsilon)$-indifferentiable from $\mathsf{R}$, if there exists an efficient (classical or quantum) simulator $\algo{S}$ and a negligible function $\epsilon$ such that, for any efficient (classical or quantum) distinguisher $\algo D$ making at most $T$ (classical or quantum) queries to $\mathsf{C}$, it holds that
$$
\left| \Pr\left[\algo D[\mathsf{C}_\lambda^{\mathrm{priv}}[\mathsf{C}_\lambda^{\mathrm{pub}}],\mathsf{C}_\lambda^{\mathrm{pub}}]=1\right] - 
\Pr\left[\algo D[\mathsf{R}_\lambda^{\mathrm{priv}},\algo{S}[\mathsf{R}_\lambda^{\mathrm{pub}}]]=1\right]\right| \leq \epsilon(\lambda).
$$
\label{def:indiff}
\end{definition}


\begin{figure}[t]
    \centering
    \includegraphics[width=.55\linewidth]{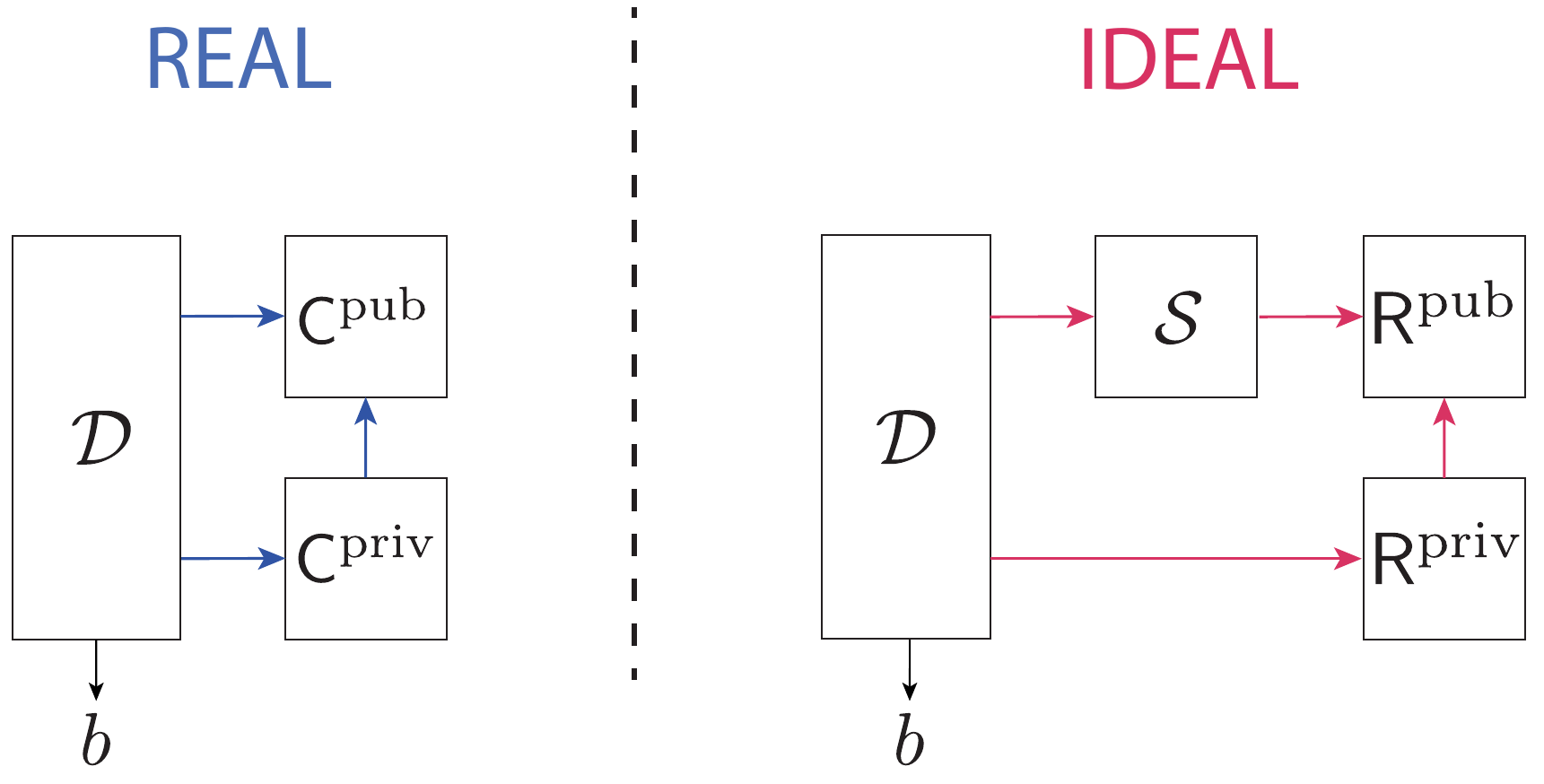}
    \caption{Schematic representation of indifferentiability of construction $\mathsf{C}$ from idealized primitive $\mathsf{R}$. The arrows denote ``access to'' the pointed system, and $\algo D$ is the distinguisher.}
    \label{fig:indiff}
\end{figure}

These notions are very general, encompassing many scenarios in which one would like to reason about building some idealized primitive $\pubpriv R$ from a different idealized primitive $\pubpriv C$ in a composably secure fashion. The composition theorem for plain indifferentiability states that a primitive instantiated using the private interface of $\pubpriv C$ is \textit{as secure as} a primitive instantiated using the private interface of $\pubpriv R$ in many scenarios. More formally, any security game that is secure against a single adversary having oracle access to $\pub R$ will also be secure against a single adversary having oracle access to $\pub C$, up to the indifferentiability loss \cite{maurer04indiff}. However, security games involving multiple rounds with distinct adversaries are not generically proven secure by indifferentiability \cite{ristenpart11careful}. For this, one requires a stronger form called \textit{reset} indifferentiability, in which the simulator $\algo S$ in Definition \ref{def:indiff} is required to be stateless.

Often, one considers security of a cryptographic system that allows a pre-processing adversary. Security games in this setting consider an adversary split into two phases: an inefficient offline phase in which some advice is computed about the underlying primitive (in our case, the interface for $\pubpriv C$ or $\pubpriv R$), and an efficient online adversary. The prepared advice is forwarded to the efficient online adversary, which then receives one or more challenge(s) that are independent of the advice. This is a more general setting than the single adversary setting of plain indifferentiability, but \textit{a priori} it may seem like security of such a game would be implied by reset indifferentiability, as this is a security game with multiple stages and adversaries. However, we show that this is not the case: this is because the offline adversary is inefficient, which is not captured by reset indifferentiability.


In fact, even strong notions of indifferentiability do not imply non-trivial space-time trade-offs. At a high level, the counterexample is a random function with a trapdoor (as both the public and private interface). Such a function is (reset, statistical, strong, quantum or classical) indifferentiable from a random oracle for query bounded adversaries, but clearly admits pre-computation attacks for one-wayness and many other security games which would be secure with a random oracle.

\section{Separating reset indifferentiability from pre-computaion}

Let $O_g$ be an oracle for a function $g: \bit^{2n} \rightarrow \bit^n$ drawn from the distribution which is uniform random, except on inputs of the form $x \Vert s$ for some random trapdoor $s \in \bit^n$. For such inputs, define $g(x \Vert s)=x$. Let oracle $O_h$ be an oracle for a function $h: \bit^{2n} \rightarrow \bit^n$ drawn uniformly at random. Note that the distributions of $g$ and $h$ have large total variation distance, but are quantum query indistinguishable by the one-way to hiding lemma \cite{unruh14revocable}. We will consider two constructions, $\pubpriv C$ and $\pubpriv R$, defined as follows. \begin{enumerate}
    \item Construction $\pubpriv C$ is defined by $\pub C := O_g$, and $\priv C [\pub C] := \pub C = O_g$. In other words, both the public and private interface for $\pubpriv C$ are the same oracle, and they are for a random function with a trapdoor.
    \item Construction $\pubpriv R$ is defined by $\pub R := O_h$, and $\priv R [\pub R] := \pub R = O_h$. In other words, both the public and private interface for $\pubpriv R$ are the same oracle, and they are for a random function (with no trapdoor).
\end{enumerate}

Now let $\algo S[\pub R]$ be the trivial simulator which has access to $h$, and answers queries $x$ as $h(x)$.

\begin{theorem}
    Construction $\pubpriv C$ is $(T, \epsilon)$ strong quantum statistical reset indifferentiable from construction $\pubpriv R$ for any $\epsilon=O(T^2/2^n)$.
\end{theorem}

\begin{proof}
    The interface for $\pubpriv R = (\pub R, \priv R[\pub R])$ is simply two oracles for the same random function. The simulated interface $(\algo S[\pub C], \priv C[\pub C])$ is two oracles for a random function, but with a random trapdoor $s$ such that inputs that end with $s$ are easy to invert. We know that $O_g$ and $O_h$ are quantum query indistinguishable in $T$ queries for any $\epsilon=O(T^2/2^n)$ \cite{unruh14revocable}, and the interface above is exactly the interface exposed to the adversary in the indistinguishability game (except two copies of each oracle are exposed; this is straightforward to simulate with one copy). This shows quantum statistical indifferentiability. The simulator can be seen to not depend on the distinguisher and be stateless, hence strong and reset.
\end{proof}

This proof straightforwardly implies the following corollary, for any choice of either bracketed item.

\begin{corollary}
    Construction $\mathsf C$ is $(T, \epsilon)$ strong $\langle$ classical | quantum $\rangle$ statistical reset indifferentiable $\langle$ with | without $\rangle$ shared randomness from construction $\mathsf R$ for any $\epsilon=O(T^2/2^n)$.
\end{corollary}

Observe that this corollary is weaker than the one shown above, as the (classical) simulator can simply ignore it's randomness. The security bound can be strengthed in the case of a classical adversary. Observe also that, given advice $s \in \bit^n$ (which depends on $\pubpriv C$, e.g. found by an unbounded adversary querying $\mathsf C$), it is straightforward to invert $g$ with no queries. However, $h$ remains hard to invert even with a much larger amount of advice, e.g. an adversary can succeed with constant probability having $S$ qubits of advice and $T$ quantum queries only when $ST+T^2 = \tilde \Omega(2^n)$ \cite{chung2020tight}.

\section{Indifferentiability with Pre-computation}
\label{section:indiff-with-pre}

Given the prior counterexample, we see that all but the strongest notions of indifferentiability are insufficient to inherit security for games that allow for adversaries with pre-computed advice (e.g. space-time tradeoff lower bounds). We define in this section a notion which captures any security game allowing unbounded pre-computation.

\paragraph{Strong indifferentiability with pre-computation.} To define the strong notion of indifferentiability with pre-computation, we will have both a pair of fixed simulators $\algo S=(\algo S_0, \algo S_1)$ and arbitrary distinguishers $\algo D=(\algo D_0, \algo D_1)$, where the first in each tuple is unbounded/offline and the second in each tuple is bounded/online. In the ``real world'', the offline distinguisher $\algo D_0$ receives unbounded access to some interface $\pubpriv C$. It then forwards $S$ (qu)bits of advice to online distinguisher $\algo D_1$, which can make $T$ queries to $\priv C$ and $\pub C$, and then outputs a bit.

In the ``ideal world'', the offline simulator $\algo S_0$ receives unbounded access to the ideal interface $\pubpriv R$, which it then uses to implement an interface which offline distinguisher $\algo D_0$ has unbounded access to. Again $\algo D_0$ forwards $S$ (qu)bits of advice to the online distinguisher $\algo D_1$, but we also allow $\algo S_0$ to forward $S_{\mathrm{sim}}$ (qu)bits of advice to the online simulator $\algo S_1$. The online distinguisher $\algo D_1$ makes $T$ queries to $\priv R$, as well as an interface simulated by $\algo S_1$ which itself makes $T_{\mathsf{sim}}$ queries to $\pub R$. As before, the distinguisher outputs a bit.

\begin{definition}[Strong Indifferentiability with Pre-Computation]
Let $\lambda \in \N$ be the security parameter.
A cryptographic system $\mathsf{C}$ is strongly $(S,T,S_{\mathsf{sim}},T_{\mathsf{sim}},\epsilon)$-indifferentiable (with pre-computation) from a system $\mathsf{R}$, if there exists a pair of simulators $(\algo{S}_0,\algo{S}_1)$, where 
\begin{itemize}
    \item $\algo S_0$ is an (classical/quantum) query and computation unbounded algorithm that outputs at most $S_{\mathsf{sim}}$ (qu)bits, and

    \item $\algo S_1$ is an efficient (classical/quantum) algorithm making $T_{\mathsf{sim}}$ (classical/quantum) queries,
\end{itemize}
and a negligible function $\epsilon(\lambda)$ 
such that, for any pair of algorithms $(\algo D_0,\algo D_1)$, where 
\begin{itemize}
    \item $\algo D_0$ is an unbounded (classical/quantum) algorithm which outputs $S$-many (qu)bits and

    \item $\algo D_1$ is an efficient (classical/quantum) making at most $T$ queries to $\mathsf{C}$,
\end{itemize}
such that the following distinguishing property holds:
\begin{align*}
\Big| &\Pr\left[\algo D_1\Big[\mathsf{C}_\lambda^{\mathrm{priv}}[\mathsf{C}_\lambda^{\mathrm{pub}}],\mathsf{C}_\lambda^{\mathrm{pub}},\algo D_0[\mathsf{C}_\lambda]\Big]=1\right] - \\
&\Pr\left[\algo D_1\Big[\mathsf{R}_\lambda^{\mathrm{priv}},\algo{S}_1\big[\mathsf{R}_\lambda^{\mathrm{pub}},\algo{S}_0[\mathsf{R}_\lambda]_{\algo S}\big],\algo D_0[\algo S_0[\mathsf{R}_\lambda]_{\algo D}]\Big]=1\right]\Big| \leq \epsilon(\lambda).
\end{align*}
Here, we assume that $\algo D_1$ only has access to the interface $D_0$ via its output, i.e., it receives $S$ many (qu)bits of advice. 
\end{definition}

The loss of the simulator, $T_{\mathsf{sim}}$ and $S_{\mathsf{sim}}$, as well as the distinguishing advantage $\epsilon$, will enter into the bounds inherited through this notion. Naturally, the smaller these quantities are, the tighter the bounds. We leave our definitions general so as to allow inheriting the tightest possible bounds.

\paragraph{Weak indifferentiability with pre-computation.} The weak indifferentiability with pre-computation game is similar to strong, except the simulator $\algo S=(\algo S_0, \algo S_1)$ now depends on distinguisher $\algo D = (\algo D_0, \algo D_1)$. This allows us to consider offline simulators $\algo S_0$ which prepare both the advice for the online simulator $\algo S_1$, as well as for the online distinguisher $\algo D_1$. This can reduce to the notion of strong indifferentiability when the offline simulator $\algo S_0$ internally runs offline distinguisher $\algo D_0$. Our definition is depicted in Figure \ref{fig:indiff-with-precomp}.

\begin{definition}[Weak Indifferentiability with Pre-Computation]
Let $\lambda \in \N$ be the security parameter.
A cryptographic system $\mathsf{C}$ is weakly $(S,T,S_{\mathsf{sim}},T_{\mathsf{sim}},\epsilon)$-indifferentiable (with pre-computation) from a system $\mathsf{R}$, if, for any pair of algorithms $(\algo D_0,\algo D_1)$, where 
\begin{itemize}
    \item $\algo D_0$ is a query and computation unbounded (classical/quantum) algorithm which outputs $S$-many (qu)bits and

    \item $\algo D_1$ is an efficient distinguisher (with binary output) making at most $T$ queries to $\mathsf{C}$,
\end{itemize}
there exists a pair of (classical/quantum) simulators $(\algo{S}_0,\algo{S}_1)$, where 
\begin{itemize}
    \item $\algo S_0$ is a query and computation unbounded (classical/quantum) algorithm that outputs $S_{\mathsf{sim}}$ (qu)bits, and

    \item $\algo S_1$ is an efficient (classical/quantum) algorithm making $T_{\mathsf{sim}}$ (classical/quantum) queries,
\end{itemize}
and a negligible function $\epsilon(\lambda)$ such that the following holds:
\begin{align*}
\Big| &\Pr\left[\algo D_1\Big[\mathsf{C}_\lambda^{\mathrm{priv}}[\mathsf{C}_\lambda^{\mathrm{pub}}],\mathsf{C}_\lambda^{\mathrm{pub}},\algo D_0[\mathsf{C}_\lambda]\Big]=1\right] - \\
&\Pr\left[\algo D_1\Big[\mathsf{R}_\lambda^{\mathrm{priv}},\algo{S}_1\big[\mathsf{R}_\lambda^{\mathrm{pub}},\algo{S}_0[\mathsf{R}_\lambda]_{\algo S}\big],\algo{S}_0[\mathsf{R}_\lambda]_{\algo D}\Big]=1\right]\Big| \leq \epsilon(\lambda).
\end{align*}
Here, we assume that $\algo D_1$ only has access to the interface $D_0$ via its output, i.e., it receives $S$ many (qu)bits of advice.
\end{definition}

\paragraph{Additional Variants.} Moreover, we consider the following variants which apply to both weak and strong indifferentiability with pre-computation:
\begin{itemize}
    \item \textbf{Computational/statistical/perfect:} \emph{Computational} indifferentiability requires the distinguisher $\algo D_1$ to be a computationally efficient algorithm. \emph{Statistical} indifferentiability requires $\algo D_1$ to make a bounded number of queries. Finally, \emph{perfect} indifferentiability refers to the case when $\algo D_1$ is completely unbounded.

    \item \textbf{Computational/statistical simulation:}
    Indifferentiability with \emph{computational simulation} requires $\algo S_1$ to be a computationally efficient simulator. Indifferentiability with \emph{statistical simulation} requires $\algo S_1$ to be only query-efficient (but computation unbounded).

    \item \textbf{Shared randomness:} In \emph{shared randomness} indifferentiability, the simulators $\algo S_0$ and $\algo S_1$ have access to the same arbitrary-sized set of random coins $\mathsf{SR} \in \bit^*$. Note that we do not reveal this randomness to the distinguisher.

    \item \textbf{Classical/quantum:} \emph{Quantum} indifferentiability captures the security against quantum distinguishers (and also simulators) making classical or quantum queries to their oracles, whereas \emph{classical} indifferentiability only considers classical algorithms that make classical queries. In the context of \emph{quantum} indifferentiability, we also distinguish between \emph{classical} and \emph{quantum} pre-computation, i.e., whether $\algo D_1$ receives classical or quantum advice from $\algo D_0$.
\end{itemize}

\begin{figure}[t]
    \centering
    \includegraphics[width=0.5\linewidth]{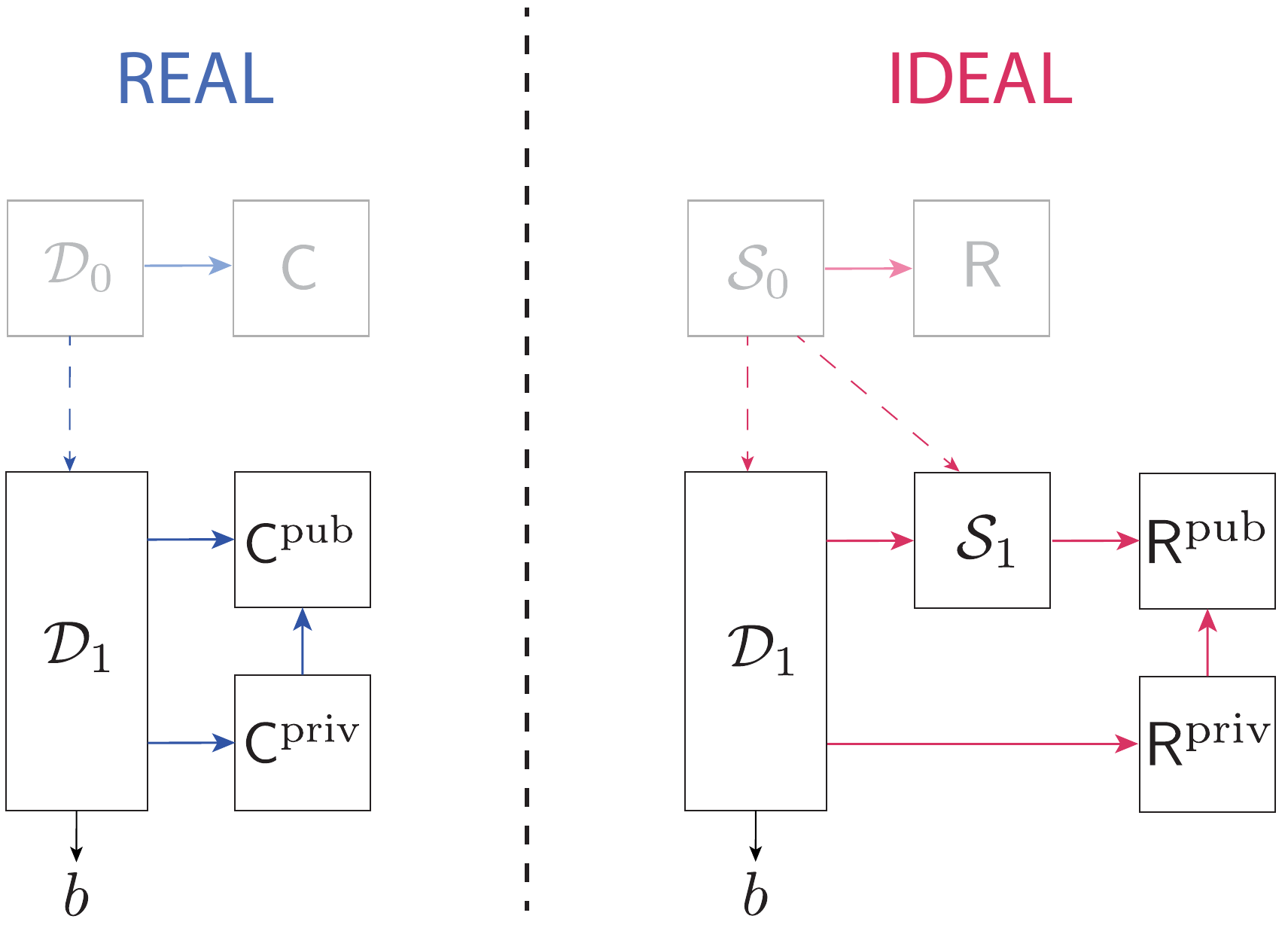}
    \caption{Schematic representation of weak indifferentiability with pre-computation. Arrows denote access to the pointed to interface, and washed out colors denote inefficient pre-computation, i.e. entities with unbounded access. Dotted arrows denote forwarded advice.}
    \label{fig:indiff-with-precomp}
\end{figure}

\subsection{Perfect reset indifferentiability suffices for pre-computation}

In this section, we observe that if two constructions $\pubpriv R$, $\pubpriv C$ are \textit{perfectly} reset indifferentiable, even for query and computation unbounded adversaries, then any multi-stage security game with (some or all) adversaries unbounded that is secure in the $\pubpriv R$ model will also be secure in the $\pubpriv C$ model. This includes our security games with pre-computation as a special case. Further, it holds even in the case where an unbounded adversary can distinguish with only negligible advantage $\epsilon$ (whereas ``perfect indifferentiability'' usually requires exactly zero advantage)---we call this notion $\epsilon$-perfect indifferentiability.

\begin{lemma}
    Suppose that construction $\pubpriv C$ is $\epsilon$-perfect reset indifferentiable from $\pubpriv R$, and any choice of the remaining variants. Suppose that the simulator makes (at most) $T_{\mathrm{sim}}$ queries to $\pub R$ to implement $T$ of the distinguisher's queries. Then $\pubpriv C$ is $\indiffparams$ indifferentiable with precomputation from $\pubpriv R$, for the same choice of remaining variants, and for $S_\mathrm{sim}=0$, $T_{\mathrm{sim}}$ and $\epsilon$ as defined above, and any number of distinguisher queries $T$ and distinguisher advice size $S$.
    \label{lem:perfect-reset-implies-precomp}
\end{lemma}

\begin{proof}
    We will prove the claim separately for strong and weak indifferentiability, though the proofs proceed analogously for any choice of the remaining variants.\begin{itemize}
        \item \emph{(Case 1)} Strong indifferentiability with pre-computation. Let $\algo S[\pub R]$ be the simulator that achieves $\epsilon$-perfect reset indifferentiability from $\pubpriv R$ (note that $\algo S$ may also be a function of some shared randomness). To construct a simulator in the strong indifferentiability with pre-computation setting $\algo S' = (\algo S_0', \algo S_1')$, we simply identify $\algo S_0'[\pubpriv R] = (\algo S[\pub R], \priv R)$ and $\algo S_1'[\pub R] = \algo S[\pub R]$. We know that $\algo S$ is stateless, so no advice needs to be passed from $\algo S_0'$ to $\algo S_1'$ while running $\algo S$. Any distinguisher $\algo D' = (\algo D_0', \algo D_1')$ in the strong indifferentiability with pre-computation game against $\algo S$ can now be identified with the distinguisher $\algo D = \algo D'$ in the strong $\epsilon$-perfect reset indifferentiability game (i.e. $\algo D$ simply runs $\algo D_0'$ followed by $\algo D_1'$, passing advice internally as needed, and outputting the result). By the reset indifferentiability of $\pubpriv C$ from $\pubpriv R$, this will distinguish with advantage at most $\epsilon$\, and hence $\algo D'$ distinguishes with advantage at most $\epsilon$.
        \item \emph{(Case 2)} Weak indifferentiability with pre-computation. Let $\algo D' = (\algo D_0', \algo D_1')$ be a distinguisher for the indifferentiability with pre-computation game. Identify $\algo D'$ with a distinguisher $\algo D$ in the $\epsilon$-perfect reset indifferentiability game of $\pubpriv C$ from $\pubpriv R$, specifically the distinguisher $\algo D$ which runs $\algo D'$ internally in two stages, forwarding advice where necessary. Let $\algo S[\pub R]$ be a stateless simulator such that $\algo D$ achieves advantage at most $\epsilon$ when run using $\algo S$.
        
        We construct $\algo S' = (\algo S_0', \algo S_1')$ in the indifferentiability with pre-computation game as $\algo S_0'[\pubpriv R]=\algo D_0'[\algo S[\pub R], \priv R]_{\algo D}$ (meaning the offline simulator runs the offline distinguisher with simulator $\algo S$ and forwards the advice created to $\algo D_1'$, and nothing to online simulator $\algo S_1'$) and $\algo S_1'[\pub R] = \algo S[\pub R]$. By the $\epsilon$-perfect reset indifferentiability achieved by $\algo S$ we know that $\algo D$ will distinguish with advantage at most $\epsilon$, which implies that $\algo D'$ will distinguish with advantage at most $\epsilon$.
    \end{itemize}
    In neither case did we place any constraints on the size of the advice from offline to online distinguisher, nor bound the number of queries of online distinguisher, so this holds for any $S, T$. The simulator overhead is clearly at most $T_{\mathrm{sim}}$, as the simulator $\algo S$ receives $T$ online queries in each case, and there is no advice passed between simulators so $S_{\mathrm{sim}}=0$.
\end{proof}

Our proof of indifferentiability with pre-computation and shared randomness of the one-round sponge from a random oracle---Section \ref{sec:sponge-indiff}---will first show that the one-round sponge is perfect reset indifferentiability  with shared randomness (a strengthening of \cite{Zhandry21}, Theorem 10), using a symmetrization argument. However, it is unclear how to remove the shared randomness from the reset indifferentiability game---in particular, techniques based on extracting randomness from an oracle such as used in \cite{Zhandry21}, do not work against an adversary that can learn the full oracle, as in our setting. Restricting to indifferentiability with pre-computation as defined in Section \ref{section:indiff-with-pre}, we next show how to remove the shared randomness from this definition, as needed for the composition theorem.

\subsection{Removing shared randomness from weak indifferentiability}

A useful fact we will prove here is that weak (or strong) indifferentiability with pre-computation, statistical simulation, and shared randomness (and any choice of the remaining variants) implies weak indifferentiability with pre-computation \textit{without} shared randomness and statistical simulation (for the same choice of remaining variants), up to a single bit of loss in the advice size.

\begin{lemma}
    Suppose that $\pubpriv C$ is weakly $\indiffparams$ indifferentiable with pre-computation from $\pubpriv R$ with statistical simulation, shared randomness, as well as with a $\langle{\text{computational | statistical | perfect}}\rangle$ $\langle{\text{classical | quantum}}\rangle$ distinguisher. Then, the construction $\pubpriv C$ is weakly $(S, T, S_{\mathrm{sim}} + 1, T_{\mathrm{sim}}, \epsilon)$ indifferentiable with precomputation from $\pubpriv R$ with statistical simulation, no shared randomness, and $\langle{\text{computational | statistical | perfect}}\rangle$ $\langle{\text{classical | quantum}}\rangle$ distinguisher (for the same choice of variants as in the premise).
    \label{lem:remove-sr}
\end{lemma}

The transformation is depicted in Figure \ref{fig:indiff-removing-sr}. 

\begin{proof}
    Let $\algo D = (\algo D_0, \algo D_1)$ be a distinguisher for the indifferentiability with pre-computation security game between $\pubpriv C$ and $\pubpriv R$, and let $\algo S = (\algo S_0, \algo S_1)$ be the simulator which witnesses the indifferentiability with pre-computation and shared randomness, incurring query loss $T_{\mathrm{sim}}$ and advice size loss $S_{\mathrm{sim}}$. Let $p$ be the probability that $\algo D$ outputs $1$ when run in the ``ideal'' world, i.e. with simulator $\algo S$ and system $\pubpriv R$. From the premise, we know that $p$ is at most $\epsilon$ away from the probability that $\algo D$ outputs $1$ when run in the ``real'' world, with no simulator and system $\pubpriv C$. We will construct an $\algo S' = (\algo S_0', \algo S_1')$ which uses no shared randomness and incurs query loss $T_{\mathrm{sim}}$ and advice size loss $S_{\mathrm{sim}} + 1$, and causes $\algo D$ to output $1$ with probability $p$ as well. Note that this suffices because any distinguisher for the game with shared randomness is equally well a distinguisher for the game without shared randomness and vice versa; the distinguisher interfaces match syntactically.

    Let us denote by $\algo S[\cdots, \mathsf{SR}]$ a simulator with shared randomness $\mathsf{SR} \in \bit^*$. We will consider simulators where this input is hard coded to value $\mathsf{SR}$ (recall that $\algo S$ is computationally unbounded, so this is valid). We then have two cases. \begin{itemize}
        \item \emph{(Case 1)} For any $\mathsf{SR} \in \bit^*$, we have \begin{align*}
            \Pr\left[\algo D_1\Big[\mathsf{R}_\lambda^{\mathrm{priv}},\algo{S}_1\big[\mathsf{R}_\lambda^{\mathrm{pub}},\algo{S}_0[\mathsf{R}_\lambda^{\mathrm{pub}}, \mathsf{SR}]_{\algo S}, \mathsf{SR}\big],\algo{S}_0[\mathsf{R}_\lambda^{\mathrm{pub}}, \mathsf{SR}]_{\algo D}\Big]=1\right] = p.
        \end{align*}
        In this case, we can simply choose any fixed value for the shared randomness $\mathsf{SR}$ and hard-code it into both offline and online simulators.
        \item \emph{(Case 2)} There are two values of shared randomness, $\mathsf{SR}_0, \mathsf{SR}_1 \in \bit^*$, such that
        \begin{align*}
            \Pr\left[\algo D_1\Big[\mathsf{R}_\lambda^{\mathrm{priv}},\algo{S}_1\big[\mathsf{R}_\lambda^{\mathrm{pub}},\algo{S}_0[\mathsf{R}_\lambda^{\mathrm{pub}}, \mathsf{SR}_0]_{\algo S}, \mathsf{SR}_0\big],\algo{S}_0[\mathsf{R}_\lambda^{\mathrm{pub}}, \mathsf{SR}_0]_{\algo D}\Big]=1\right] =& p_0 \\
            \Pr\left[\algo D_1\Big[\mathsf{R}_\lambda^{\mathrm{priv}},\algo{S}_1\big[\mathsf{R}_\lambda^{\mathrm{pub}},\algo{S}_0[\mathsf{R}_\lambda^{\mathrm{pub}}, \mathsf{SR}_1]_{\algo S}, \mathsf{SR}_1\big],\algo{S}_0[\mathsf{R}_\lambda^{\mathrm{pub}}, \mathsf{SR}_1]_{\algo D}\Big]=1\right] =& p_1
        \end{align*}
        and further $p_0 < p < p_1$. In this case, we hard-code both $\mathsf{SR}_0$ and $\mathsf{SR}_1$ into the offline and online simulators $\algo S_0', \algo S_1'$. Then, before running any other computation, $\algo S_0'$ samples a bit $s \in \bit$ such that \begin{align*}
            \Pr_{\algo S'}[s=1] =& \frac{p - p_0}{p_1-p_0}.
        \end{align*}
        After selecting bit $s$, the constructed offline simulator $\algo S_0'$ simply runs $\algo S[\cdots, {\mathsf{SR}_s}]$, i.e. the initial simulator with shared randomness hard-coded as dictated by $s$. In addition to whatever advice $\algo S_0$ would send, $\algo S_0'$ also sends the bit $s$ to $\algo S_1'$. Then, $\algo S_1'$ runs $\algo S_1[\cdots, {\mathsf{SR}_s}]$, again hard-coding the shared randomness as dictated by $s$. This clearly incurs only an overhead of one (qu)bit of advice. We now have \begin{align*}
            &\Pr\left[\algo D_1\Big[\mathsf{R}_\lambda^{\mathrm{priv}},\algo{S}_1^{'}\big[\mathsf{R}_\lambda^{\mathrm{pub}},\algo{S}_0^{'}[\mathsf{R}_\lambda^{\mathrm{pub}}]_{\algo S}\big],\algo{S}_0^{'}[\mathsf{R}_\lambda^{\mathrm{pub}}]_{\algo D}\Big]=1\right]\\
            &= p_0 \cdot \Pr_{\algo S'}[s=0] + p_1 \cdot \Pr_{\algo S'}[s=1] 
            = p,
        \end{align*}
        which proves the claim.
    \end{itemize}
\end{proof}

\begin{figure}
    \centering
    \includegraphics[width=0.6\linewidth]{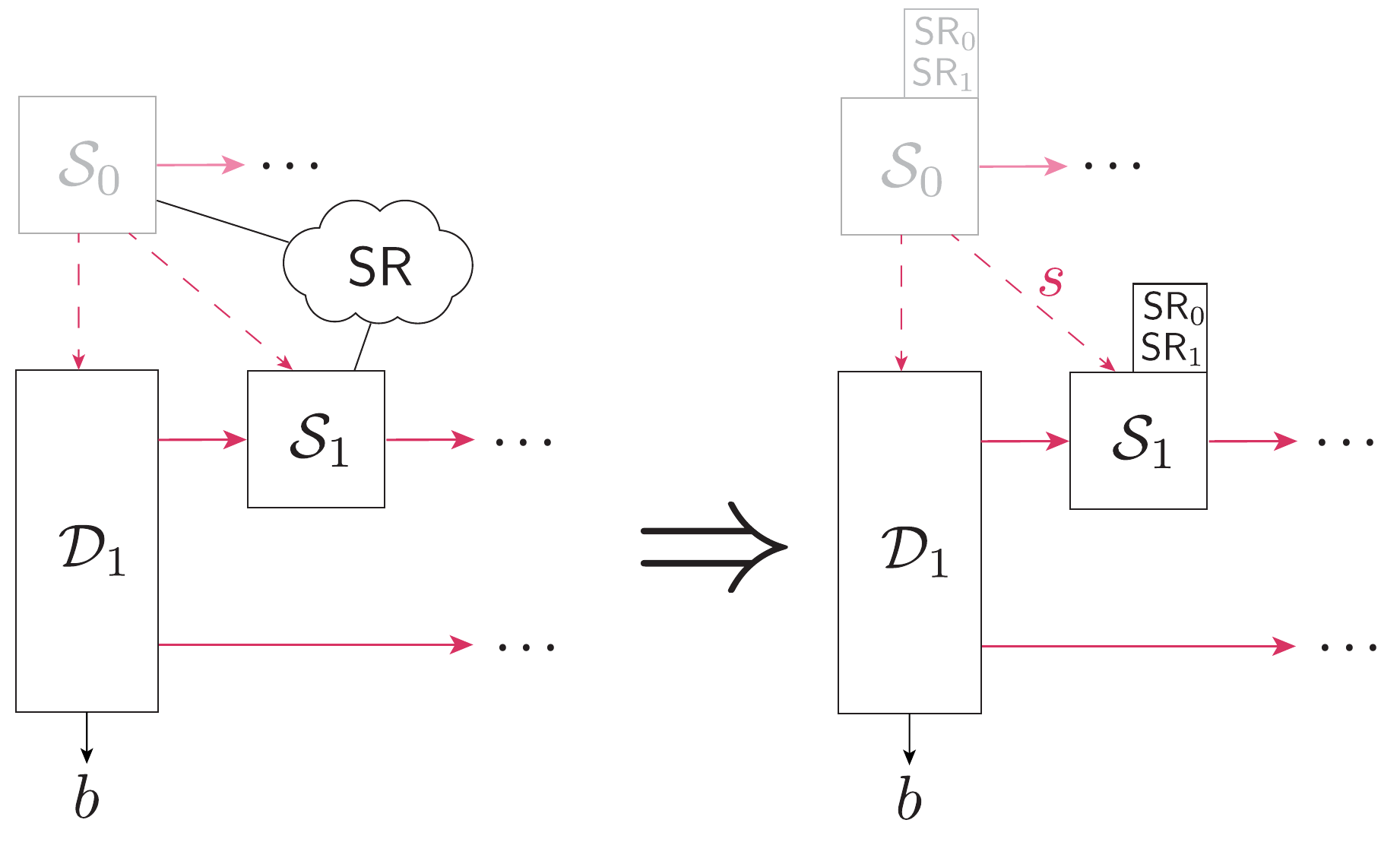}
    \caption{The reduction for removing shared randomness for weak indifferentiability. Two values of shared randomness are hard-coded into the simulator, which then uses bit $s$ to select between them.}
    \label{fig:indiff-removing-sr}
\end{figure}

\section{Composition Framework for Indifferentiability with Pre-Computation}
\label{sec:composition}

We show here that our proposed notion of indifferentiability with pre-computation implies composed security for a class of security games allowing pre-computing adversaries. This class of games includes the setting of most space-time tradeoffs as a special case.

\subsection{Security games with pre-computation}

For the remainder of this section, all objects will in fact be a family indexed by a security parameter $\lambda$; we drop this index for convenience. We will follow the indifferentiability model of Maurer et al. \cite{maurer04indiff}. In particular, let $\pubpriv C$ be a construction with a private interface $\priv C$ and a public interface $\pub C$ as defined in Section \ref{section:indiff-with-pre}. The components of a generic security game with pre-computation will be a cryptosystem $\algo P$, an enviroment $\algo E$, a pre-computation adversary $\algo A_0$, and an online adversary $\algo A_1$. For simplicity, we will quantify efficiency only in terms of number of queries to the interface $\pubpriv C$; all entities will be computationally unbounded. We further state our results for quantum queries and advice with weak indifferentiability, though one could define an analogous framework for classical queries and/or classical advice and/or quantum or classical computationally bounded entities and/or strong indifferentiability. The composition theorem proceeds similarly for all such cases. We do not consider composition theorems with shared randomness, as we will not use them.

In the following definitions, when we say that two objects ``interact'', we mean that they alternatively send data back and forth for some number of rounds. The number and order of rounds, as well as size and type (classical or quantum) of the data is a property of a specific security game. We will leave our definition general enough to capture any suitable security game. The only parameters we will need in a general game are the advice size $S$, and online query count $T$ (which counts queries made by $\algo A_1$ and cryptosystem $\algo P$).

\begin{definition}
    An offline adversary $\algo A_0$ is a (computational and query) unbounded quantum algorithm, which interacts with $\pubpriv C$ (both $\priv C$ and $\pub C$) for an arbitrary number of rounds, and prepares an advice state $\alpha$ that is $S$ qubits.
    \label{defn:offline-adv}
\end{definition}

\begin{definition}
    An online adversary $\algo A_1$ is an interactive quantum algorithm, which makes $T_1$ queries to the public interface $\pub C$, and interacts with the environment $\algo E$ and cryptosystem $\algo P$.
    \label{defn:online-adv}
\end{definition}

\begin{definition}
    A cryptosystem $\algo P$ is an interactive quantum algorithm, which makes $T_2$ queries to the private interface $\priv C$, and interacts with the environment $\algo E$ and online adversary $\algo A_1$.
    \label{defn:cryptosystem}
\end{definition}

\begin{definition}
    An environment $\algo E$ is an interactive quantum algorithm, which interacts with the cryptosystem $\algo P$ and online adversary $\algo A_1$. At the end of the experiment, $\algo E$ outputs a bit $b$.
    \label{defn:environment}
\end{definition}
In Definition \ref{defn:online-adv} and Definition \ref{defn:cryptosystem}, we require that $T=T_1+T_2$. We call the tuple $(\algo P, \pubpriv C, \algo E)$ an instance of the $\algo P$ cryptosystem in model $\pubpriv C$. We can also consider a different interface $\pubpriv R$ with a private interface $\priv R$ that syntactically matches $\priv C$ (e.g. if $\priv C$ is an oracle for a function from $\bit^* \rightarrow \bit^n$, then $\priv R$ is as well). We call the tuple $(\algo P, \pubpriv R, \algo E)$ an instance of the $\algo P$ cryptosystem in model $\pubpriv R$. We are now ready to state our composition theorem.

\begin{theorem}
    Suppose that construction $\pubpriv C$ is $(S, T, S_\mathsf{sim}, T_\mathsf{sim}, \epsilon)$ indifferentiable with pre-computation from construction $\pubpriv R$. Let $\algo A = (\algo A_0, \algo A_1)$ be an attacker in the $\pubpriv C$ model of $\algo P$ with advice size $S$ and online query count $T_1$, in a game where $\algo P$ makes $T_2$ queries to $\priv C$ such that $T_1+T_2 = T$. Then there is an attacker $\algo A' = (\algo A_0', \algo A_1')$ in the $\pubpriv R$ model of $\algo P$ with advice size $S+S_\mathsf{sim}$ and online query count $T_\mathsf{sim}$. This attacker satisfies
    \begin{align*}
        \Big|\Pr\left[\algo E\left[\algo P[\priv C], \algo A_1[\pub C, \algo A_0[\pubpriv C]]\right] = 1\right] -&\\ \Pr\left[\algo E\left[\algo P[\priv R], \algo A_1'[\pub R, \algo A_0'[\pubpriv R]]\right] = 1\right]\Big| &\leq \epsilon
    \end{align*}
    \label{thm:composition}
\end{theorem}

The construction of $\algo A$ is depicted in Figure \ref{fig:composition-theorem-constructed-adversary}, and the proof of correctness is depicted in Figure \ref{fig:composition-theorem-proof}.

\begin{proof}
    Let $\algo S[\pubpriv R] = (\algo S_0[\pubpriv R], \algo S_1[\priv R])$ be a simulator which is $(S, T, S_\mathsf{sim}, T_\mathsf{sim}, \epsilon)$-indifferentiable in the indifferentiability with pre-computation game against interface $\pubpriv C$, for a certain distinguisher $\algo D = (\algo D_0, \algo D_1)$ which will be defined later. We can construct adversary $\algo A'$ in the $\pubpriv R$ model of $\algo P$ from the adversary $\algo A$ in the $\pubpriv C$ model of $\algo P$ and the simulator $\algo S$ in a black-box way. In particular, the pre-processing adversary $\algo A_0'$ is simply $\algo S_0[\priv R]$, which prepares advice state $\alpha_{\algo A}$ of size $S$ and $\alpha_{\algo S}$ of size $S_\mathsf{sim}$. The online adversary is then the joint system of the online adversary and online simulator, $\algo A_1' = \algo A_1[S_1[\priv R, \alpha_{\algo S}], \alpha_{\algo A}]$. It is clear that this adversary $\algo A'$ uses $S+S_\mathsf{sim}$ qubits of advice and $T_\mathsf{sim}$ quantum queries (note that the original $T_1$ queries by the adversary are now queries to the simulator, and not to the construction; hence we only need count $T_\mathsf{sim}$).
    
    To show the success probability gap, let $\delta_1$ be the probability that the environment $\algo E$ outputs $1$ in the $\pubpriv C$ model of $\algo P$ with adversary $\algo A$, and let $\delta_2$ be the probability for the same environment $\algo E$ outputting $1$ in the $\pubpriv R$ model of $\algo P$ with adversary $\algo A'$. Note that the tuple $(\algo E, \algo P, \algo A_1)$ is a valid online distinguisher for the indifferentiability with pre-computation game, which we call $\algo D_1$. Similarly, $\algo A_0$ is a valid preprocessing distinguisher, which we call $\algo D_0$. Hence we have a distinguisher $\algo D = (\algo D_0, \algo D_1)$ for the indifferentiability with pre-computation game---this distinguisher $\algo D$ is the one which defines $\algo S$ earlier. Further, it will output $1$ with probability $\delta_1$ in the ideal world (from the definition of our security game), and with probability $\delta_2$ in the real world (from the construction of $\algo A'$). Hence, we must have $|\delta_1 - \delta_2| \leq \epsilon$, demonstrating the claim.
\end{proof}

\begin{figure}[h]
    \centering
    \includegraphics[width=0.3\linewidth]{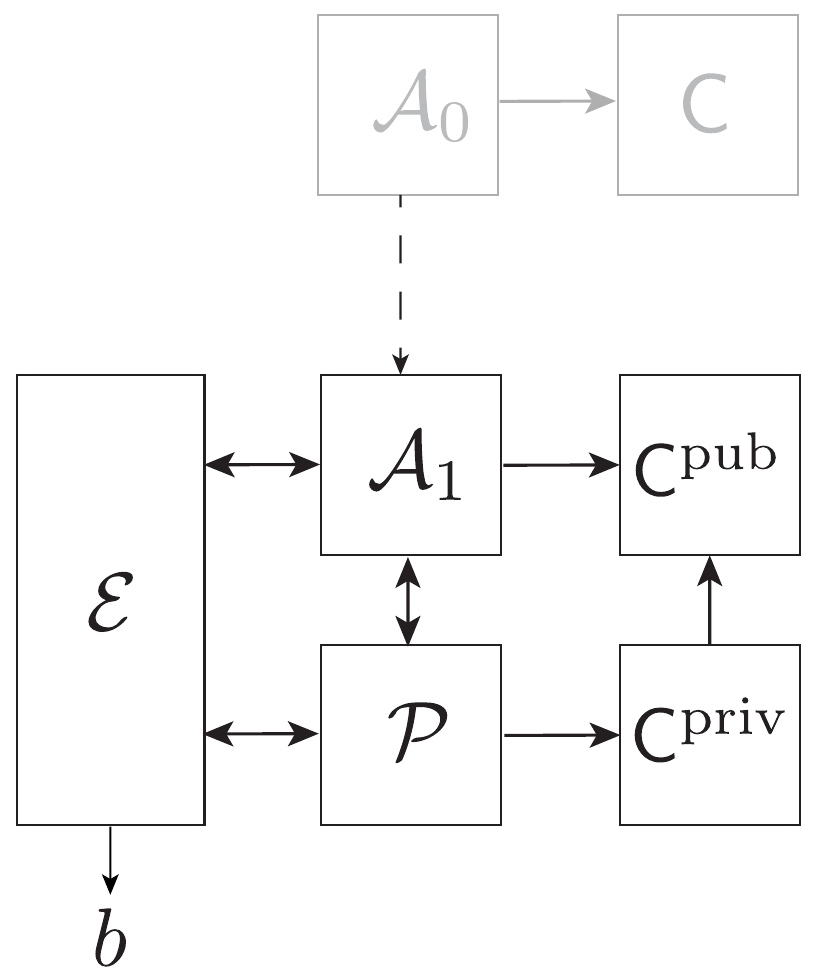}
    \caption{The adversary $\algo A=(\algo A_0, \algo A_1)$ in model $\pubpriv C$.}
    \label{fig:composition-theorem-base-adversary}
\end{figure}

\begin{figure}[h]
    \centering
    \includegraphics[width=0.4\linewidth]{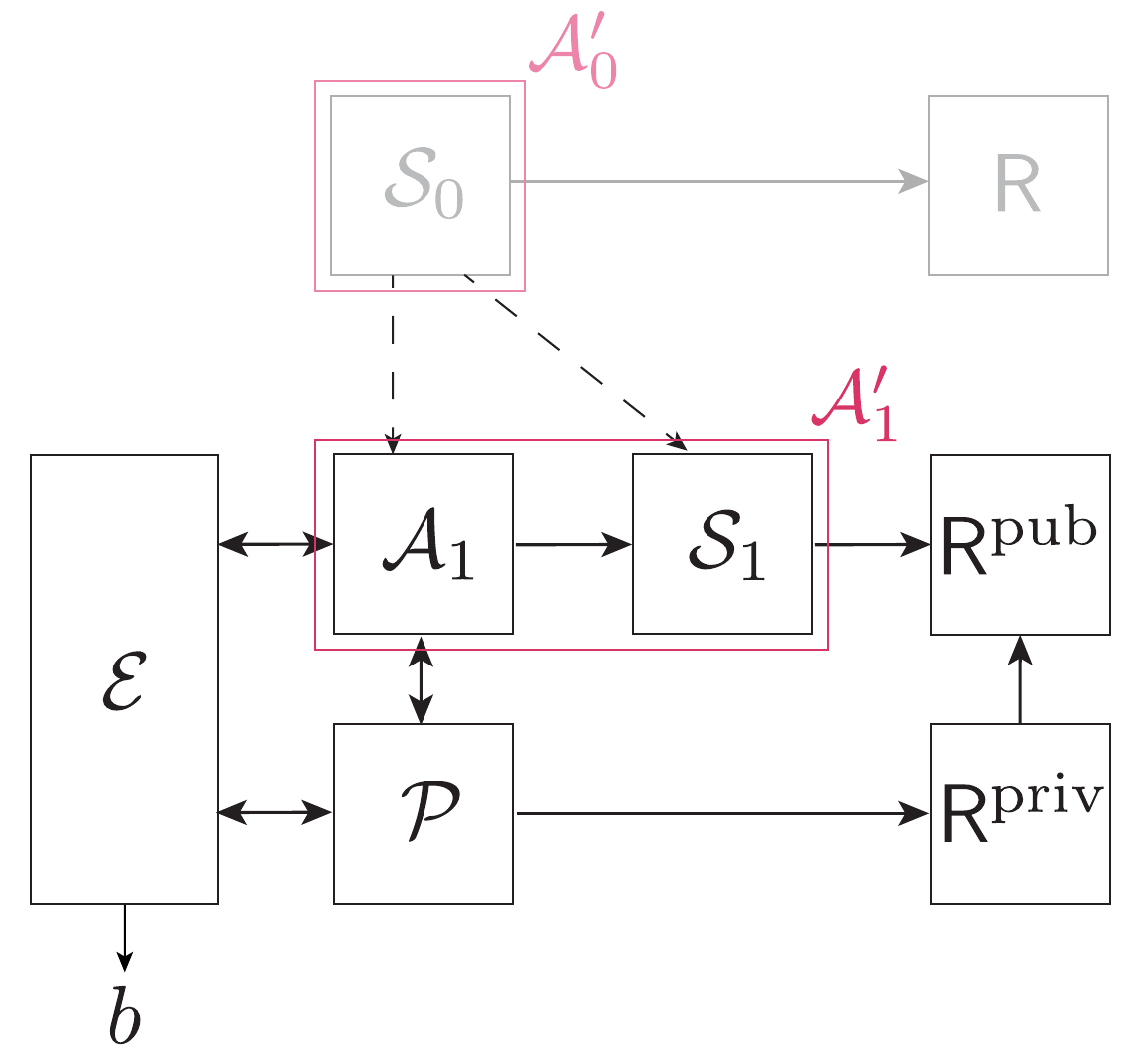}
    \caption{The adversary $\algo A'=(\algo A_0', \algo A_1')$ in model $\pubpriv R$, constructed from $\algo A$ in model $\pubpriv C$ and the simulator $\algo S$. The syntax of $\algo A$ is depicted in Figure \ref{fig:composition-theorem-base-adversary}.}
    \label{fig:composition-theorem-constructed-adversary}
\end{figure}

\begin{figure}[h]
    \centering
    \includegraphics[width=0.65\linewidth]{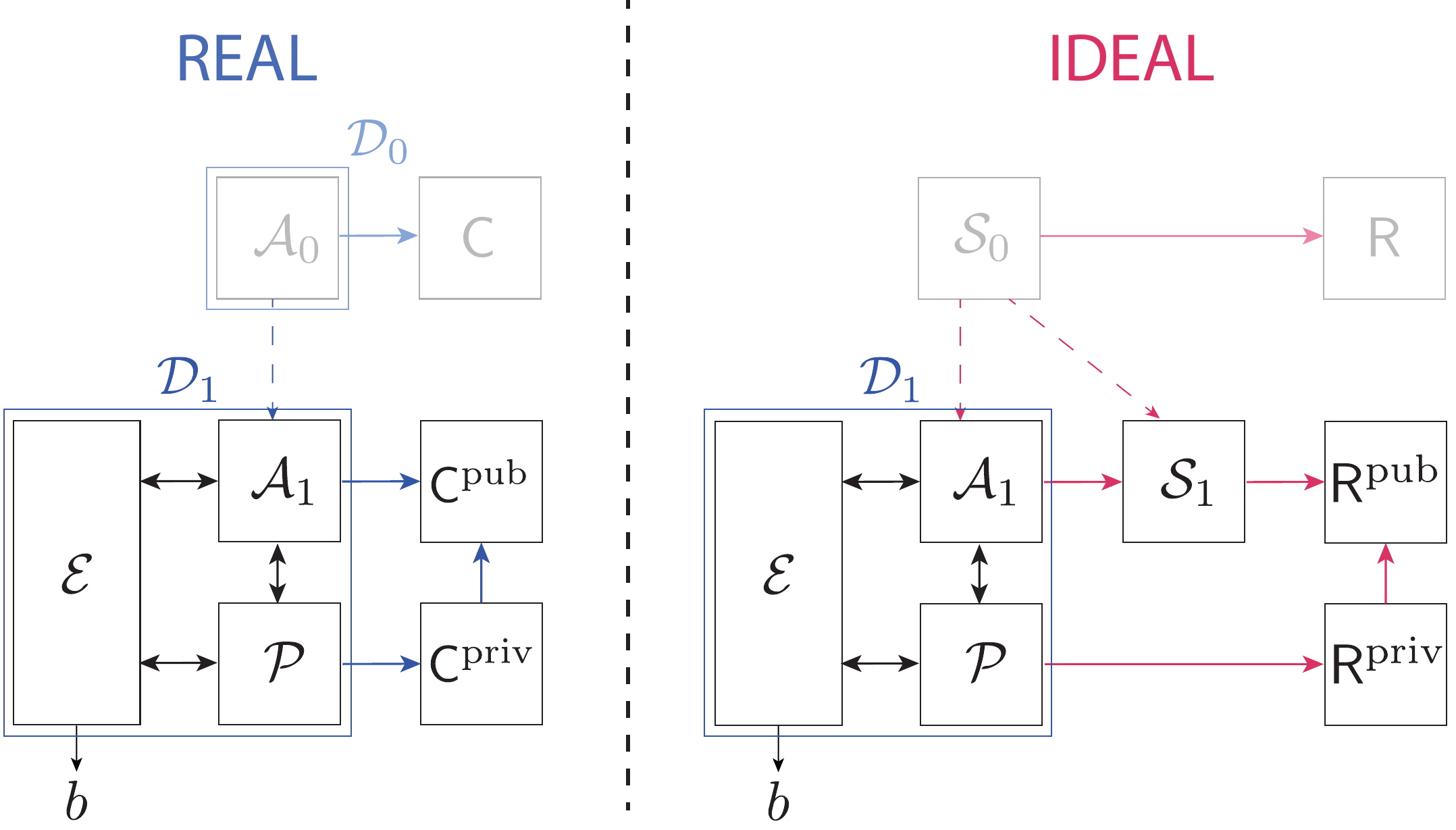}
    \caption{The reduction from the constructed adversary in the $\pubpriv R$ model to a distinguisher $\algo D = (\algo D_0, \algo D_1)$ for the indifferentiability game between $\pubpriv C$ and $\pubpriv R$.}
    \label{fig:composition-theorem-proof}
\end{figure}

\section{Sponge symmetrization}
\label{sec:improved-sym}

In this section we introduce the relevant background building up to our symmetrization lemma. This lemma will be the main technical component in showing that the single round sponge is indifferentiable with pre-computation from a random oracle, when the rate is smaller than the capacity.

\subsection{Group theory}

We first define relevant notions in group theory, and recall known results about the symmetric group $S_N$ on $N$ elements. For a more complete overview of the subject, we refer the reader to the work of James \cite{James1984}. This presentation follows \cite{carolan2024oneway}.


Let $S_N$ denote the symmetric group consisting of permutations which act on the set $[N]:=\{1, ...,N\}$. For a subset $A \subset [N]$, let $S_A$ denote the maximal subgroup of $S_N$ which fixes every element in the complement of $A$, i.e. $[N] \setminus A$. Now let $A_1,...,A_l$ be a partition of $[N]$ such that the disjoint union satisfies $\bigsqcup_{i \in [\ell]} A_i =[N]$.

\begin{definition}[Young subgroup]
    A subgroup $H$ of the symmetric group $S_N$ is a Young subgroup if it can be expressed as $H=S_{A_1} \times ... \times S_{A_l}$, where $\times$ denotes the internal direct product and the collection of subsets $\{A_i\}_{i \in [\ell]}$ forms a partition of $[N]$. \label{definition:youngSubgroup}
\end{definition}

The concept of a double coset, which we review below, will also be relevant.

\begin{definition}[Double cosets] Let $H,K$ be subgroups of a group $G$. The double cosets of $G$ under $(H, K)$, denoted $H \diagdown G \diagup K$, are the sets of elements which are invariant under left multiplication by $H$ and right multiplication by $K$. In particular, 
$$
H \diagdown G \diagup K = \big\{\{hxk\, : \,h \in H, k \in K\}\, : \,x \in G\big\}.$$ \label{definition:doubleCoset}
\end{definition}
\vspace{-2mm}
It is well-known that $G$ is the disjoint union of its double cosets for any subgroups $H,K \leq G$. We focus on the double cosets of the symmetric group $S_N$, specifically those generated by Young subgroups. These subgroups admit the following characterization, adapted from Jones \cite{Jones1996DoubleCosets} and James \cite{James1984}.

\begin{theorem}[\cite{Jones1996DoubleCosets}, Theorem 2.2]
    Let $H, K$ be Young subgroups of $S_N$, with corresponding partitions $A_1,...,A_l$ (for $H$) and $B_1,...,B_m$ (for $K$). Let $\pi \in S_N$ and $C=H\pi K$ be the corresponding double coset. Any other permutation $\pi' \in S_N$ is in $C$ if and only if for all $i \leq l, j \leq m$ we have $|A_i \cap \pi' B_j| = |A_i \cap \pi B_j|$.  \label{theorem:characterizationYoungDoubleCosets}
\end{theorem}

Intuitively, the above characterization says that the double cosets defined by Young subgroups $(H, K)$ correspond to sets of permutations which look the same if one only considers how they distribute the elements of each $B_j$ among the different $A_i$'s. Two permutations $\pi, \pi'$ are in the same double coset if and only if for every $A_i, B_j$ both $\pi$ and $\pi'$ send the same number of elements from $B_j$ to $A_i$. This characterization combined with the following lemma is a key component of our reduction.

\begin{theorem}[\cite{wildonSymGroupRep}, Theorem 4.4] For any subgroups $H, K$ of a (finite) group $G$ with $x, g \in G$ both in the same $(H, K)$ double coset, there are exactly $|x^{-1}Hx \cap K|$ ways of choosing $h \in H$ and $k \in K$ such that $g=h x k$. \label{lemma:cardinalityDoubleCosets}
\end{theorem}

As a corollary, it follows that selecting random symmetrizing elements from $H$ and $K$ suffice to give a random element of the double coset.

\begin{lemma}
    For any subgroups $H, K$ of a (finite) group $G$ with $x \in G$, let $h \sim H$ and $k \sim K$ be uniform random. Then $g=hxk$ is uniform random over the double coset $HxK$. \label{lemma:symmetrization}
\end{lemma}

\begin{proof}
    Fix some $g \in HxK$ from the double coset of $x$. The number of ways to choose an $h \in H$ and a $k \in K$ such that $hxk=g$ is independent of $g$; hence for any $g$, the probability of obtaining $g=hxk$ is the same.
\end{proof}

\subsection{Symmetrization lemma}

In this section, we construct a symmetrized permutation $\varphi$ from a random function $f$ such that the sponge hash of $\varphi$ exactly matches the function $f$ (at least when $r \leq c$), and $\varphi$ is exponentially close to uniform random. Furthermore, a query to $\varphi$ or $\varphi^{-1}$ can be implemented with a single query to $f$. At a high level, we pick symmetrizing permutations such that the double cosets consist exactly of permutations with the same sponge hash. We use the notation $n=r+c$, and $\lambda = \min(r, c)$. We assume $r \leq c$ which implies $r=\lambda$. 

\begin{lemma}[Symmetrization of the one-round sponge]
    Let $f : \bit^r \rightarrow \bit^r$ be a random function. Let $\algo C \subset \algo S_{2^n}$ be a subset of permutations on $n$-bit strings where $n=r+c$ and $r \leq c$, such that $\varphi \in \algo C$ if and only if $\Sp^{\varphi}=f$. Then there exists a (quantum or classical) algorithm that samples a random $\varphi \sim \algo C$, and can implement queries to $\varphi$ and $\varphi^{-1}$ each with a single query to $f$. Further, $\varphi$ cannot be distinguished from a random permutation with advantage greater than $O(2^{-r/2})$.
    \label{lem:sponge-sym}
\end{lemma}

\begin{proof}
Define the transversal permutation $\pi_f$ in terms of the random function $f: \bit^{r} \rightarrow \bit^r$. Then, for $x \in \bit^r$, $y \in \bit^r$, and $g\in \bit^{n-2r}$ (observing that $g$ may be the empty string)
\begin{align}
    \pi_f(x \Vert g \Vert y) =& y \oplus f(x) \Vert g \Vert  x.
\end{align}
This construction is such that $\Sp^{\pi_f} = f$, e.g. the functions have the same truth table. Note also that both $\pi_f$ and $\pi_f^{-1}$ can be implemented using a single query to $f$. To construct the right symmetrizing subgroup, define the singleton sets $B_z$ for all $z \in \bit^r$ as \begin{align}
    B_z := \{z \Vert 0^c\},
\end{align}
and define $B_\bot$ as the complement of all $B_z$, e.g. \begin{align}
    B_\bot := \bit^n \setminus \cup_{z \in \bit^r} B_z.
\end{align}
We now observe that the sets $\{B_z\}_{z \in \bit^{r} \cup \{\bot\}}$ define a partition of $\bit^n$. Thus, we can define the corresponding right symmetrizing Young subgroup $K \leq S_{2^{r+c}}$ as \begin{align}
    K :=& \big\{\sigma \in S_{2^{r+c}} \, : \, \sigma(B_z) = B_z, \,\forall z \in \bit^{r} \cup \{\bot\}\big\}.
\end{align}

For the left symmetrizing subgroup, define the sets $A_x \subset \bit^n$ as \begin{align}
    A_x :=& \{x \Vert y \, : \, y \in \bit^{c}\},
\end{align}
which are easily seen to partition the set $\bit^n$. The $\{A_x\}_{x \in \bit^r}$ sets therefore define the left symmetrizing Young subgroup $H \leq S_{2^{r+c}}$ with \begin{align}
    H :=& \big\{\omega \, : \, \omega(A_x) = A_x, \forall x \in \bit^{r}\big\}.
\end{align}
Suppose we now sample $\omega \sim H$ and $\sigma \sim K$ uniformly at random and symmetrize $\pi_f$ to create a new permutation $\varphi$. In particular, we let
\begin{align}
    \omega \sim& H, && \sigma \sim K, &&
    \varphi := \omega \circ \pi_f \circ \sigma.
\end{align}

\begin{remark}
    If we let $G := S_{2^{n}}$, then the double cosets $H \diagdown G \diagup K$ are exactly sets of permutations which have the same sponge hash. In particular, $\pi, \pi'$ are in the same double coset if and only if $\Sp^\pi = \Sp^{\pi'}$
    \label{rem:sponge-hash-double-cosets}
\end{remark}

To see the above, recall the generic characterization that two permutations are in the same young double cosets if they distribute the elements of each $B_i$ among the $A_j$ in exactly the same way, Theorem \ref{theorem:characterizationYoungDoubleCosets}. There are two directions to show. \begin{itemize}
    \item[$(\rightarrow)$] Suppose $\pi, \pi'$ define the same sponge hash function. Then for any $z \in \bit^r$, we have $\pi(z \Vert 0^c)[:r] = \pi'(z \Vert 0^c)[:r]$, and hence if $\pi(B_z) \subset A_i$ then $\pi'(B_z) \subset A_i$. The $B_z$ are singleton sets, so this gives the equation \begin{align}
        |\pi(B_z) \cap A_x| = |\pi'(B_z) \cap A_x|, \text{ for all $x, z \in \bit^r$.}
    \end{align} The distribution of $B_z$ determine the distribution of $B_\bot$ as \begin{align}
        |\pi(B_\bot) \cap A_x| = |A_x| - \sum_{z \in \bit^r} |\pi(B_z) \cap A_x|
    \end{align}
    and similarly for $\pi'$. Therefore, the equality holds for all $x, z$ including $z=\bot$.
    \item[$(\leftarrow)$] Similar to the above argument, but in reverse; the implications go both directions.
\end{itemize}

We have from Remark \ref{rem:sponge-hash-double-cosets} above and the symmetrization lemma (Lemma \ref{lemma:symmetrization}) that $\varphi$ is uniform random over all permutations having the same sponge hash function as $\pi$. From Lemma \ref{lem:sponge-ttable-random} in Appendix \ref{app:tech-lemmas}, we have that the sponge hash of $\pi$ can be distinguished from the distribution induced by the sponge hash of a truly random permutation with advantage $O(2^{-r/2})$, even given the full truth table of the sponge. It follows from our symmetrization argument that $\varphi$ can be distinguished from a uniform random permutation with advantage $O(2^{-r/2})$, even given the whole truth table.
\end{proof}

\section{Indifferentiability (with Pre-Computation) of the One-Round Sponge}
\label{sec:sponge-indiff}

In this section we show that the one-round sponge is $\epsilon$-perfect\footnote{This means that query unbounded adversaries have advantage at most $\epsilon$} strongly reset indifferentiable with shared randomness, for exponentially small $\epsilon$, both quantumly and classically. This strengthens a result of \cite{Zhandry21}, which proved statistical instead of perfect indifferentiability, i.e. security only against query bounded adversaries. As a corollary, we then have that the one-round sponge is strongly indifferentiable with pre-computation and shared randomness from Lemma \ref{lem:perfect-reset-implies-precomp}. We can then remove the shared randomness from this definition at the cost of switching to weak indifferentiability with computationally unbounded simulator through Lemma \ref{lem:remove-sr}. With this result in place, we illustrate how our composition theorem implies a tight space-time tradeoff for the one-round sponge inversion. More broadly, our composition theorem allows the one-round sponge to inherit any space-time tradeoff of a random oracle.

\subsection{Proof of Indifferentiability}

\begin{theorem} Let $r,c \in \N$ with $r \leq c$ and $\lambda = \min(r,c)$ be parameters. Let $\varphi:\bit^{r+c} \rightarrow \bit^{r+c}$ be a random permutation. Then, the (one-round) sponge construction $\mathsf{C}$ with $$
\mathsf{C}_\lambda^{\mathrm{priv}} = \Sp^\varphi \quad\quad \text{and} \quad\quad \mathsf{C}_\lambda^{\mathrm{pub}} = (\varphi,\varphi^{-1})$$ 
is strong, $\langle\text{ classical  }\vert\text{ quantum 
 }\rangle$, is $\epsilon$-perfect reset indifferentiable (shared randomness) from a random oracle $\pubpriv R$, where for a random $f: \bit^{r} \rightarrow \bit^r$, the interfaces $\priv R = \pub R$ correspond to an oracle for $f$, and where $\epsilon = O(2^{-r/2})$.
\label{thm:sponge-indiff}
\end{theorem}

\begin{proof}
We will prove the result for quantum distinguisher, though the proof proceeds similarly for classical. Let $\algo D$ be an unbounded quantum algorithm which makes queries to $\pubpriv C$. Consider the following sequence of hybrid experiments:

\begin{description}
    \item $\mathbf{Game \, 1:}$ This hybrid corresponds to the \emph{real world}. The adversary $\algo D$ 
    receives access to the (single-round) sponge construction $\mathsf{C}$ with 
    $$
\mathsf{C}_\lambda^{\mathrm{priv}} = \Sp^\varphi \quad \text{ and } \quad \mathsf{C}_\lambda^{\mathrm{pub}} = (\varphi,\varphi^{-1})$$
where $\varphi:\bit^{r+c} \rightarrow \bit^{r+c}$ is a random permutation and $r\leq c$. Define the event
$$
\mathbf{Game \, 1} := \left[ b=1 \, : \, b \leftarrow \algo D[\Sp^\varphi,(\varphi,\varphi^{-1})] \right].
$$

\item $\mathbf{Game \, 2:}$ This corresponds to the following intermediate experiment. The adversary $\algo D$ receives access to the (single-round) sponge construction $\mathsf{C}$ with 
    $$
\mathsf{C}_\lambda^{\mathrm{priv}} = \Sp^{\hat\varphi} \quad \text{ and } \quad \mathsf{C}_\lambda^{\mathrm{pub}} = (\hat\varphi,\hat\varphi^{-1})$$
where $\hat\varphi:\bit^{r+c} \rightarrow \bit^{r+c}$ is a permutation which is generated as follows:
\begin{enumerate}
    \item Sample a uniformly random function $f: \bit^r \rightarrow \bit^r$.

    \item Run $\hat{\varphi} \leftarrow \mathsf{Sym}_{r,c}(f)$ using the symmetrization procedure in Algorithm \ref{alg:sym-r-c-f}.
\end{enumerate}

We define the corresponding hybrid event by
$$
\mathbf{Game \, 2} := \left[ b=1 \, : \, b \leftarrow \algo D[\Sp^{\hat\varphi},(\hat\varphi,\hat\varphi^{-1})] \right].
$$

\item $\mathbf{Game \, 3:}$ This hybrid corresponds to the \emph{ideal world}. The adversary $\algo D$ receives access to interfaces which are simulated by the following stateless simulator $\algo S$
 which has access to a random oracle $f: \bit^{r} \rightarrow \bit^r$ and a common source of shared randomness $\mathsf{SR} \in \bit^*$. This is the simulator for the interface $(\hat\varphi,\hat\varphi^{-1})$: it answers queries using the permutation $\hat\varphi \leftarrow \mathsf{Sim}_{r,c}^{f}(\cdot \,;\mathsf{SR})$ which can be evaluated via oracle calls to the random oracle $f$ and shared randomness $\mathsf{SR} \in \bit^*$, where $\mathsf{Sim}_{r,c}^{f}(\cdot \,;\mathsf{SR})$ is the procedure in Algorithm \ref{alg:sim-f} which internally calls $f$.

We define the corresponding event for the ideal world by
$$
\mathbf{Game \, 3} := \Big[ b=1 \, : \, b \leftarrow \algo D[f,\algo S[f,\mathsf{SR}]] \Big].
$$
\end{description}

\begin{algorithm}[t]
\DontPrintSemicolon
\SetAlgoLined
\label{alg:sym-r-c-f}
\KwIn{Parameters $r,c \in \N$ and a truth table for a function $f: \bit^r \rightarrow \bit^r$.}
    
\KwOut{Truth table for a permutation $\hat{\varphi}:\bit^{r+c} \rightarrow \bit^{r+c}$.}

Let $\pi: \bit^{r+c} \rightarrow \bit^{r+c}$ be the permutation with
$$
\pi(x \Vert g \Vert y) := y \oplus f(x) \Vert g \Vert  x.
$$

For $z \in \bit^r$, let $B_z = \{z \Vert 0^c\}$ and $B_\bot =\bit^n \setminus \cup_{z \in \bit^r} B_z$;

For $x \in \bit^r$, let $A_x = \{(x \Vert y) \, : \, y \in \bit^{c}\}$;

Sample a random permutation $\sigma$ from the Young subgroup $K \leq S_{2^{r+c}}$ with
$$
K=\{\sigma \in S_{2^{r+c}} \, : \, \sigma(B_z) = B_z, \forall z \in \bit^{r} \cup \{\bot\}\};$$

Sample a random permutation $\omega$ from the Young subgroup $H \leq S_{2^{r+c}}$ with
$$
H = \{\omega \in S_{2^{r+c}} \, : \, \omega(A_x) = A_x, \forall x \in \bit^{\lambda}\};$$

Output a truth table for the permutation $\hat\varphi = \omega \circ \pi \circ \sigma$.

\caption{$\mathsf{Sym}_{r,c}(f)$}
\end{algorithm}

First, we show the following:
\begin{claim}
$$
\left|\Pr\left[\mathbf{Game \, 2}\right] - \Pr\left[\mathbf{Game \, 1}\right]\right| \leq O(2^{-r/2}).
$$
\end{claim}
The claim follows from sponge symmetrization, Lemma \ref{lem:sponge-sym}. In particular, it was shown that distinguishing the truth table of a uniform random $\varphi$ from the truth table of a symmetrized $\varphi$ constructed from a random $f:\bit^r \rightarrow \bit^r$ can be done with advantage at most $O(2^{-r/2})$. The permutation $\varphi$ uniquely determines $\Sp^\varphi$ in both $\mathbf{Game \, 1}$ and $\mathbf{Game \, 2}$, so this suffices to prove the claim. Finally, we observe the following:
\begin{claim}
$$
\Pr\left[\mathbf{Game \, 3}\right]= \Pr\left[\mathbf{Game \, 2}\right].
$$
\end{claim}
\begin{proof}
Note that in the previous experiment, $\mathbf{Game \, 2}$, we already introduced a perfectly random function, which is now featured as a random oracle. To prove the claim, it suffices to argue that a sufficiently long shared random strong $\mathsf{SR} \in \bit^*$ enables the of simulator $\algo S$ to run Algorithm \ref{alg:sim-f} to generate the same symmetrizing permutations $\sigma \sim K$ and $\omega \sim H$ on each query (despite being stateless). Define $N =2^{r+c}$ and recall that $H,K \leq S_{N}$ are both Young subgroups. Let $B = \{B_z\}_{{z \in \bit^{r} \cup \{\bot\}}}$ denote the invariant sets with respect to $K$, and let $A = \{A_x\}_{x \in \bit^{r}}$ denote the invariant sets for $H$. Then, 
$$
K \cong \prod_{A_i \in A} A_i \quad\quad \text{ and } \quad\quad H \cong \prod_{B_j \in B} B_j.
$$
Hence, the claim is essentially just a consequence of the \emph{coupon collector problem}\footnote{In fact, to generate a random permutation $\pi \in S_N$ only $O(N \log N)$ random bits suffice on average. The probability of failure can be further suppressed with additional amounts of randomness. Since the shared random string $\mathsf{SR} \in \bit^*$ can in principle be unbounded, we did not analyze the length explicitly.}.
\end{proof}

Putting everything together, we get that
\begin{align*}
\Big|\Pr\left[\algo D[\Sp^\varphi,(\varphi,\varphi^{-1})] =1 \right]
- \Pr\Big[\algo D\big[f,\algo S[f,\mathsf{SR}]]\big]=1\Big]\Big| \, \leq \, O(2^{-r/2}).
\end{align*}
\end{proof}

We can lift $\epsilon$-perfect reset indifferentiability to the analogous notion of indifferentiability with pre-computation, as in the following corollary.

\begin{corollary}
    Let $r,c \in \N$ with $r \leq c$ and $\lambda = \min(r,c)$ be parameters. Let $\varphi:\bit^{r+c} \rightarrow \bit^{r+c}$ be a random permutation. Then, the (one-round) sponge construction $\mathsf{C}$ with $$
    \mathsf{C}_\lambda^{\mathrm{priv}} = \Sp^\varphi \quad\quad \text{and} \quad\quad \mathsf{C}_\lambda^{\mathrm{pub}} = (\varphi,\varphi^{-1})$$ 
    is strong, $\langle\text{ classical }\vert\text{ quantum }\rangle$, and perfect $(S,T,S_{\mathsf{sim}},T_{\mathsf{sim}},\epsilon)$-indifferentiable (with pre-computation and shared randomness) from a random oracle $f: \bit^{r} \rightarrow \bit^r$ for any parameters $S$ and $T$, where $S_{\mathsf{sim}}=0$ and $T_{\mathsf{sim}}=T$, and where $\epsilon = O(2^{-r/2})$.
    \label{cor:sponge-indiff-pre}
\end{corollary}

\begin{proof}
    Follows from Theorem \ref{thm:sponge-indiff} and Lemma \ref{lem:perfect-reset-implies-precomp}.
\end{proof}

\begin{algorithm}[t]
\DontPrintSemicolon
\SetAlgoLined
\label{alg:sim-f}
\KwIn{Parameters $r,c \in \N$, an oracle for a function $f: \bit^r \rightarrow \bit^r$, an input string $x_{\mathsf{in}} \in \bit^{r+c}$, and a random string $\mathsf{SR} \in \bit^*$.}
    
\KwOut{An output string $y_{\mathsf{out}} \in \bit^{r+c}$.}

For $z \in \bit^r$, let $B_z = \{z \Vert 0^c\}$ and $B_\bot = \bit^n \setminus \cup_{z \in \bit^r} B_z$;

For $x \in \bit^r$, let $A_x = \{(x \Vert y) \, : \, y \in \bit^{c}\}$;

Use a subset of the random coins in $\mathsf{SR} \in \bit^*$ to assign a random permutation $\sigma$ from the Young subgroup $K \leq S_{2^{r+c}}$ with
$$
K = \{\sigma \in S_{2^{r+c}} \, : \, \sigma(B_z) = B_z, \forall z \in \bit^{r} \cup \{\bot\}\};$$

Use another subset of the random coins in $\mathsf{SR} \in \bit^*$ to assign a random permutation $\omega$ from the Young subgroup $H \leq S_{2^{r+c}}$ with
$$
H = \{\omega \in S_{2^{r+c}} \, : \, \omega(A_x) = A_x, \forall x \in \bit^{\lambda}\};$$

Output $y_{\mathsf{out}}=\hat\varphi(x_{\mathsf{in}})$, where $\hat\varphi$ is the symmetrized permutation $\hat\varphi = \omega \circ \pi_f \circ \sigma$ and where $\pi_f: \bit^{r+c} \rightarrow \bit^{r+c}$ can be evaluated with an oracle call to $f$ via
$$
\pi_f(x \Vert g \Vert y) := y \oplus f(x) \Vert g \Vert  x.
$$

\caption{$\mathsf{Sim}_{r,c}^{f}(x_{\mathsf{in}}\,;\mathsf{SR})$}
\end{algorithm}

This further implies weak indifferentiability with shared randomness and a statistical simulator. Using Lemma \ref{lem:remove-sr}, we can lift this to weak indifferentiability with pre-computation and without shared randomness at the cost of a single bit of loss in advice size, with the remaining variants set the same way.

\begin{corollary} Let $r,c \in \N$ and $r\leq c$, with  $\lambda = r$ be parameters. Let $\varphi:\bit^{r+c} \rightarrow \bit^{r+c}$ be a random permutation. Then, the (one-round) sponge construction $\mathsf{C}$ with $$
\mathsf{C}_\lambda^{\mathrm{priv}} = \Sp^\varphi \quad\quad \text{and} \quad\quad \mathsf{C}_\lambda^{\mathrm{pub}} = (\varphi,\varphi^{-1})$$ 
is weak, $\langle\text{ classical } \vert \text{ quantum }\rangle$ and perfect $(S,T,S_{\mathsf{sim}},T_{\mathsf{sim}},\epsilon)$-indifferentiable with pre-computation and a statistical simulator, but \emph{not} shared randomness, from a random oracle $f: \bit^{r} \rightarrow \bit^r$ for any parameters $S$ and $T$, where $S_{\mathsf{sim}}=1$, $T_{\mathsf{sim}}=T$, and $\epsilon = O(2^{-r/2})$.
\label{cor:weak-indiff-sponge}    
\end{corollary}
\begin{proof}
    Follows from Corollary \ref{cor:sponge-indiff-pre} and Lemma \ref{lem:remove-sr}.
\end{proof}

\subsection{Space-Time Trade-Offs for Sponge Inversion}\label{sec:ST-trade-off}

In this section we illustrate how a (tight) quantum space-time trade-off for single-round sponge inversion is a special case of our composition theorem. In other words, we consider $\Sp^\varphi: \bit^{r} \rightarrow \bit^r$, where $\varphi: \bit^{r+c} \rightarrow \bit^{r+c}$ is a random permutation, in the case of non-uniform quantum adversaries that make $T$ quantum queries to $\varphi, \varphi^{-1}$ and take $S$ qubits of quantum advice (which may depend arbitrarily on $\varphi$). We will here assume that $r \leq c$. Chung et al.~\cite{chung2020tight} show that any quantum algorithm which finds a pre-image of a randomly generated image $f(x)$ with probability $\epsilon$ using $S$ qubits of advice, where $f$ is a random function $f: [N] \rightarrow [M]$, satisfies a space-time trade-off
$$
\epsilon \leq \tilde{O} \left(\sqrt[3]{\frac{S\cdot T + T^2}{\min(N,M)}}\right).
$$
If $S$ is classical advice, then~\cite{chung2020tight} manage to get a slightly better bound, namely
$$
\epsilon \leq \tilde{O} \left(\frac{S\cdot T + T^2}{\min(N,M)}\right).
$$

From the characterization in Corollary \ref{cor:weak-indiff-sponge}, we can simply apply Theorem \ref{thm:composition}, our composition theorem, to obtain a time-space tradeoff for the one-round sponge by lifting a result which is known for random oracles. Note that this applies for both classical and quantum space-time tradeoffs, and both classical and quantum advice. We will illustrate this procedure for sponge inversion, though it applies more generally.

\paragraph{Sponge inversion.}

We prove the following space-time trade-off relations for one-round sponge inversion.

\begin{theorem}[Space-time trade-off for sponge inversion]
Let $r,c \in \N$ be integers such that $r \leq c$.
Any (classical or quantum) inverter $\algo A = (\algo A_0,\algo A_1)$ for the one-round sponge construction $\Sp^\varphi: \bit^{r} \rightarrow \bit^r$ which consists of a pair of algorithms, where
\begin{itemize}
    \item $\algo A_0$ prepares $S$ (qu)bits of advice $\alpha$ (depending arbitrarily on $\varphi$), 
    
    \item $\algo A_1$ receives $\alpha$ and makes $T$ (classical/quantum) queries to either $\varphi$ or $\varphi^{-1}$,
\end{itemize}
 and where $\algo A$ succeeds with non-trivial\footnote{We remark that this requirement in our theorem statement can be relaxed at the cost of including an additional additive term of $O(2^{-r/2})$ in each of the space-time trade-offs.} probability $\epsilon = \omega(2^{-r/2})$ for a random permutation $\varphi: \bit^{r+c} \rightarrow \bit^{r+c}$, must obey the space-time trade-offs:
\begin{itemize}
    \item \emph{(Classical advice and queries:)}
$$
\epsilon \leq O\left(\frac{ST}{2^r}\right)
$$
    \item \emph{(Classical advice and quantum queries:)}
$$
\epsilon \leq \tilde{O} \left(\frac{S T + T^2}{2^{r}}\right).
$$
\item \emph{(Quantum advice and queries:)}
$$
\epsilon \leq \tilde{O} \left(\sqrt[3]{\frac{S T + T^2}{2^{r}}}\right).
$$
\end{itemize}
\end{theorem}

\begin{proof}
    Note that these bounds
    are all known to hold for inverting a random function $f : \bit^r \rightarrow \bit^r$. The first is due to Yao \cite{yao1990coherent} and De et al. \cite{de2010timespace}, and the second and third due to Chung et al \cite{chung2020tight}. We will show how the function inversion game can be modelled as a security game with pre-computation, as defined in Section \ref{sec:composition}. 
    
    Let us consider quantum queries and quantum advice, though the proof proceeds similarly for all three cases.
    
\begin{description}
    \item $\mathbf{Function \,\, inversion.}$ Parties will receive oracle access to an idealized random function $\mathsf{R}$ with 
    $$
    \priv R = f \quad \text{ and } \quad \pub R = f$$
    where $f:\bit^{r} \rightarrow \bit^{r}$ is a random function. The adversary is denoted $\algo A=(\algo A_0,\algo A_1)$. The offline adversary $\algo A_0$ receives unbounded access to $\pubpriv R$, and the online adversary receives access to $\pub R$. The cryptosystem $\algo P$ has access to the function through $\priv R$, and the environment $\algo E$ interacts with $\algo P$ (it will not need to interact with $\algo A$ in this game). The game proceeds as follows. \begin{enumerate}
        \item Offline adversary $\algo A_0$ receives unbounded access to $\pubpriv R$, which it uses to prepare an $S$ qubit state $\ket{\alpha_{\algo A}}$, which is forwarded to $\algo A_1$.
        \item Cryptosystem $\algo P$ samples a random $x \sim \bit^r$, and computes $y=f(x)$ using private interface $\priv R$. It forwards $y$ to $\algo A_1$.
        \item Online adversary $\algo A_1$ returns some $x' \in \bit^r$ to $\algo P$. The cryptosystem $\algo P$ then checks whether $f(x')=y$, puts the result in a bit $b \in  \bit$.
        \item Cryptosystem $\algo P$ forwards $b$ to the environment $\algo E$, which then outputs $b$.
    \end{enumerate}
    \item $\mathbf{Sponge \,\, inversion.}$ Parties will receive oracle access to the (single-round) sponge construction $\mathsf{C}$ with 
    $$
    \priv C = \Sp^\varphi \quad \text{ and } \quad \pub C = (\varphi,\varphi^{-1})$$
    where $\varphi:\bit^{r+c} \rightarrow \bit^{r+c}$ is a random permutation. The adversary is denoted $\algo A=(\algo A_0,\algo A_1)$. The offline adversary $\algo A_0$ receives unbounded access to $\pubpriv C$, and the online adversary receives access to $\pub C$. The cryptosystem $\algo P$ has access to the sponge hash $\priv C$, and the environment $\algo E$ interacts with $\algo P$ (it will not need to interact with $\algo A$ in this game). The game proceeds as follows. \begin{enumerate}
        \item Offline adversary $\algo A_0$ receives unbounded access to $\pubpriv C$, which it uses to prepare an $S$ qubit state $\ket{\alpha_{\algo A}}$, which is forwarded to online adversary $\algo A_1$.
        \item Cryptosystem $\algo P$ samples a random $x \sim \bit^r$, and computes $y=\Sp^{\varphi}(x)$ using private interface $\priv C$. It forwards $y$ to $\algo A_1$.
        \item Online adversary $\algo A_1$ returns some $x' \in \bit^r$ to $\algo P$. The cryptosystem $\algo P$ then checks whether $\Sp^{\varphi}(x')=y$, puts the result in a bit $b \in  \bit$.
        \item Cryptosystem $\algo P$ forwards $b$ to the environment $\algo E$, which then outputs $b$.
    \end{enumerate}
\end{description}
Observe that both games are instances of a security game with pre-computation as described in Section \ref{sec:composition}. Furthermore, the cryptosystem $\mathcal P$ and environment $\mathcal E$ are the same in each game, hence these two are the same game, with the first in the $\pubpriv R$ model and the second in the $\pubpriv C$ model. The cryptosystem $\mathcal P$ makes $T_2 = 2$ queries to the private interface. Now suppose that an $S, T$ adversary wins the game \textbf{Sponge inversion} (i.e. inversion as defined by $\mathcal P, \mathcal E$ in model $\pubpriv C$) with probability $\epsilon$. It follows from Theorem \ref{thm:composition} that there is an $S+1, T+2$ adversary winning \textbf{Function inversion} (i.e. inversion as defined by $\mathcal P, \mathcal E$ in model $\pubpriv R$) with probability at least $\epsilon-O(2^{-r/2})$. This proves the claim for quantum advice and queries. A similar line of reasoning, using the appropriate notion of indifferentiability with pre-computation, proves the remaining two claims.

\end{proof}

Note that significantly improving these bounds would imply new circuit lower bounds, by a result of Corrigan-Gibbs and Kogan \cite{corrigan2019function}. In particular, note that a time-space tradeoff lower bound for sponge inversion implies a similar time-space tradeoff lower bound for function inversion.


%
%
%
\bibliographystyle{alpha}
\bibliography{ref}
%
 

\newpage
\appendix
\section{Technical lemmas}
\label{app:tech-lemmas}
In this section we will show that the full truth table of $\Sp^\varphi$ for a random permutation $\varphi$, and the full truth table of a random function from $r$ bits to $r$ bits, can be distinguished with only exponentially small probability. This holds even when given a sample that is the entire truth table of these two functions; this will be necessary to show that our symmetrization procedure is sound.

\begin{lemma}
    Let $\algo D_1$ be the uniform distribution on functions from $\bit^r \rightarrow \bit^r$. Let $\algo D_2$ be the distribution on functions from $\bit^r$ to $\bit^r$ induced by sampling $\varphi:\bit^n\rightarrow\bit^n$ uniformly at random and taking $\Sp^\varphi$. Then the maximum distinguishing advantage of an algorithm $\algo A$ given a sample (i.e. a full truth table) from $\algo D_1$ or from $\algo D_2$ satisfies \begin{align*}
        \left|\Pr_{f \sim D_1} [\algo A(f) = 1] - \Pr_{f \sim D_2} [\algo A(f) = 1]\right| \leq O\big(2^{-\lambda/2}\big).
    \end{align*}
    \label{lem:sponge-ttable-random}
\end{lemma}

\begin{proof}
    We begin with an alternative characterization of $\algo D_1$. Consider drawing a random $f':\bit^{n} \rightarrow \bit^r$, and then defining $f:\bit^r \rightarrow \bit^r$ as $f(x):=f'(x||0^c)$; $f$ is the sample. The restriction of a random function is still a random function, so this characterization is equivalent, i.e. $f$ is a uniform random function from $r$ bits to $r$ bits. 

    Now we will use Lemma \ref{lem:trunc-perm-advantage} to upper bound the maximum distinguishing advantage given a sample of one of the two truth tables. We can build a (classical) adversary for the game defined in Lemma \ref{lem:trunc-perm-advantage} which queries the $2^r$ inputs of the form $x||0^c$ to obtain a truth table of size $2^r\times r$. We can information-theoretically distinguish a truth table that came from a truly random function ($f'$ above, note that in this case we have the truth table of $f$) from those that come from a truncated permutation ($\varphi$ above, note that in this case we have the truth table of $\Sp^\varphi$) with advantage $\mathsf{ADV}$. We use $q=2^r$ classical queries to the extended function/truncated permutation $f/\varphi$ to construct the truncated function, and have $m=n-\lambda$, so from Lemma \ref{lem:trunc-perm-advantage} we have
    \begin{align*}
        \mathsf{ADV} = O\left(\frac{2^r}{2^\frac{2n-\lambda}{2}}\right)= O\left(2^{-\lambda/2}\right).
    \end{align*}
\end{proof}

We have used the following lemma from Gilboa and Gueron.

\begin{lemma}[\cite{gilboa21truncate}]
    A (classical) adversary which makes $q$ queries to either \begin{enumerate}
        \item a random function $f:\bit^n \rightarrow \bit^{n-m}$, or
        \item a random permutation $\varphi:\bit^n \rightarrow \bit^n$ where the last $m$ bits are discarded (i.e. not learned)
    \end{enumerate}
    can distinguish with advantage $O\left(\frac{q}{2^\frac{n+m}{2}}\right)$.
    \label{lem:trunc-perm-advantage}
\end{lemma}

\end{document}